\documentclass[11pt,oneside,english,american,reqno]{amsart}
\usepackage[T1]{fontenc}
\usepackage{verbatim}
\usepackage{amsbsy}
\usepackage{amstext}
\usepackage{amsthm}
\usepackage{amssymb}
\usepackage{undertilde}
\usepackage{xargs}[2008/03/08]

\makeatletter
\numberwithin{equation}{section}
\theoremstyle{plain}
\newtheorem{thm}{\protect\theoremname}
  \theoremstyle{definition}
  \newtheorem{defn}[thm]{\protect\definitionname}
  \theoremstyle{definition}
  \newtheorem{example}[thm]{\protect\examplename}

\usepackage{pdfsync} 
\usepackage[all]{xy}
\newcommand{\ie}{\textit{i.e.}}
\newcommand{\eg}{\textit{e.g.}}
\usepackage[lining,tabular]{fbb} 
\usepackage[libertine,bigdelims]{newtxmath}
\usepackage[cal=boondoxo,bb=boondox,frak=boondox]{mathalfa}

   \usepackage{bm}
\usepackage{bm}
\renewcommand{\mathbf}{\bm}
 \theoremstyle{plain}
  \newtheorem{proposition}[thm]{Proposition}
\usepackage{xcolor} 
\definecolor{brown(traditional)}{rgb}{0.59, 0.29, 0.0}
\definecolor{blue(ryb)}{rgb}{0.01, 0.28, 1.0}
\definecolor{red}{rgb}{1.0, 0.0, 0.0}
\definecolor{magenta}{rgb}{1.0, 0.0, 1.0}
\definecolor{mahogany}{rgb}{0.75, 0.25, 0.0}
\definecolor{lavenderpurple}{rgb}{0.59, 0.48, 0.71}
\definecolor{olive}{rgb}{0.5, 0.5, 0.0}
\definecolor{brickred}{rgb}{0.8, 0.25, 0.33}
\definecolor{antiquefuchsia}{rgb}{0.57, 0.36, 0.51}
\definecolor{bole}{rgb}{0.47, 0.27, 0.23}
\definecolor{darkolivegreen}{rgb}{0.33, 0.42, 0.18}
\definecolor{deepjunglegreen}{rgb}{0.0, 0.29, 0.29}

\makeatother

\usepackage{babel}
  \addto\captionsamerican{\renewcommand{\definitionname}{Definition}}
  \addto\captionsamerican{\renewcommand{\examplename}{Example}}
  \addto\captionsamerican{\renewcommand{\theoremname}{Theorem}}
  \addto\captionsenglish{\renewcommand{\definitionname}{Definition}}
  \addto\captionsenglish{\renewcommand{\examplename}{Example}}
  \addto\captionsenglish{\renewcommand{\theoremname}{Theorem}}
  \providecommand{\definitionname}{Definition}
  \providecommand{\examplename}{Example}
\providecommand{\theoremname}{Theorem}

\begin{document}

\global\long\def\ga{\alpha}
\global\long\def\gb{\beta}
\global\long\def\ggm{\gamma}
\global\long\def\go{\omega}
\global\long\def\gs{\sigma}
\global\long\def\gd{\delta}
\global\long\def\gD{\Delta}
\global\long\def\vph{\varphi}
\global\long\def\gf{\varphi}
\global\long\def\gk{\kappa}
\global\long\def\gl{\lambda}
\global\long\def\gz{\zeta}
\global\long\def\gh{\eta}
\global\long\def\gy{\upsilon}

\global\long\def\eps{\varepsilon}
\global\long\def\epss#1#2{\varepsilon_{#2}^{#1}}
\global\long\def\ep#1{\eps_{#1}}

\global\long\def\wh#1{\widehat{#1}}

\global\long\def\spec#1{\textsf{#1}}

\global\long\def\ui{\wh{\boldsymbol{\imath}}}
\global\long\def\uj{\wh{\boldsymbol{\jmath}}}
\global\long\def\uk{\widehat{\boldsymbol{k}}}

\global\long\def\uI{\widehat{\mathbf{I}}}
\global\long\def\uJ{\widehat{\mathbf{J}}}
\global\long\def\uK{\widehat{\mathbf{K}}}

\global\long\def\bs#1{\boldsymbol{#1}}
\global\long\def\vect#1{\mathbf{#1}}
\global\long\def\bi#1{\textbf{\emph{#1}}}

\global\long\def\uv#1{\widehat{\boldsymbol{#1}}}
\global\long\def\cross{\times}

\global\long\def\ddt{\frac{\dee}{\dee t}}
\global\long\def\dbyd#1{\frac{\dee}{\dee#1}}
\global\long\def\dby#1#2{\frac{\partial#1}{\partial#2}}

\global\long\def\vct#1{\mathbf{#1}}

\global\long\def\partialby#1#2{\frac{\partial#1}{\partial x^{#2}}}
\newcommandx\parder[2][usedefault, addprefix=\global, 1=]{\frac{\partial#2}{\partial#1}}

\global\long\def\fall{,\quad\text{for all}\quad}

\global\long\def\reals{\mathbb{R}}

\global\long\def\rthree{\reals^{3}}
\global\long\def\rsix{\reals^{6}}
\global\long\def\rn{\reals^{n}}
\global\long\def\rt#1{\reals^{#1}}

\global\long\def\les{\leqslant}
\global\long\def\ges{\geqslant}

\global\long\def\dee{\textrm{d}}
\global\long\def\di{d}

\global\long\def\from{\colon}
\global\long\def\tto{\longrightarrow}
\global\long\def\lmt{\longmapsto}

\global\long\def\abs#1{\left|#1\right|}

\global\long\def\isom{\cong}

\global\long\def\comp{\circ}

\global\long\def\cl#1{\overline{#1}}

\global\long\def\fun{\varphi}

\global\long\def\interior{\textrm{Int}\,}

\global\long\def\sign{\textrm{sign}\,}
\global\long\def\sgn#1{(-1)^{#1}}
\global\long\def\sgnp#1{(-1)^{\abs{#1}}}

\global\long\def\dimension{\textrm{dim}\,}

\global\long\def\esssup{\textrm{ess}\,\sup}

\global\long\def\ess{\textrm{{ess}}}

\global\long\def\kernel{\mathop{\textrm{\textup{Kernel}}}}

\global\long\def\support{\mathop{\textrm{\textup{support}}}}

\global\long\def\image{\mathop{\textrm{\textup{Image}}}}

\global\long\def\diver{\mathop{\textrm{\textup{div}}}}

\global\long\def\sp{\mathop{\textrm{\textup{span}}}}

\global\long\def\resto#1{|_{#1}}
\global\long\def\incl{\iota}
\global\long\def\iden{\imath}
\global\long\def\idnt{\textrm{Id}}
\global\long\def\rest{\rho}
\global\long\def\extnd{e_{0}}

\global\long\def\proj{\textrm{pr}}

\global\long\def\L#1{L\bigl(#1\bigr)}
\global\long\def\LS#1{L_{S}\bigl(#1\bigr)}

\global\long\def\ino#1{\int_{#1}}

\global\long\def\half{\frac{1}{2}}
\global\long\def\shalf{{\scriptstyle \half}}
\global\long\def\third{\frac{1}{3}}

\global\long\def\empt{\varnothing}

\global\long\def\paren#1{\left(#1\right)}
\global\long\def\bigp#1{\bigl(#1\bigr)}
\global\long\def\biggp#1{\biggl(#1\biggr)}
\global\long\def\Bigp#1{\Bigl(#1\Bigr)}

\global\long\def\braces#1{\left\{  #1\right\}  }
\global\long\def\sqbr#1{\left[#1\right]}
\global\long\def\anglep#1{\left\langle #1\right\rangle }

\global\long\def\lsum{{\textstyle \sum}}

\global\long\def\bigabs#1{\bigl|#1\bigr|}

\global\long\def\stp{\text{\small\ensuremath{\bigodot}}}
\global\long\def\tp{\text{\small\ensuremath{\bigotimes}}}

\global\long\def\mi#1{\boldsymbol{#1}}
\global\long\def\mii{\mi I}
\global\long\def\mie#1#2{#1_{1}\cdots#1_{#2}}
\global\long\def\asmi#1{#1}
\global\long\def\ordr#1{\left\langle #1\right\rangle }

\global\long\def\symm#1{\paren{#1}}
\global\long\def\smtr{\mathcal{S}}

\global\long\def\perm{p}
\global\long\def\sperm{\mathcal{P}}

\global\long\def\oneto{1,\dots,}

\global\long\def\lisub#1#2#3{#1_{1}#2\dots#2#1_{#3}}

\global\long\def\lisup#1#2#3{#1^{1}#2\dots#2#1^{#3}}

\global\long\def\lisubb#1#2#3#4{#1_{#2}#3\dots#3#1_{#4}}

\global\long\def\lisubbc#1#2#3#4{#1_{#2}#3\cdots#3#1_{#4}}

\global\long\def\lisubbwout#1#2#3#4#5{#1_{#2}#3\dots#3\widehat{#1}_{#5}#3\dots#3#1_{#4}}

\global\long\def\lisubc#1#2#3{#1_{1}#2\cdots#2#1_{#3}}

\global\long\def\lisupc#1#2#3{#1^{1}#2\cdots#2#1^{#3}}

\global\long\def\lisupp#1#2#3#4{#1^{#2}#3\dots#3#1^{#4}}

\global\long\def\lisuppc#1#2#3#4{#1^{#2}#3\cdots#3#1^{#4}}

\global\long\def\lisuppwout#1#2#3#4#5#6{#1^{#2}#3#4#3\wh{#1^{#6}}#3#4#3#1^{#5}}

\global\long\def\lisubbwout#1#2#3#4#5#6{#1_{#2}#3#4#3\wh{#1}_{#6}#3#4#3#1_{#5}}

\global\long\def\lisubwout#1#2#3#4{#1_{1}#2\dots#2\widehat{#1}_{#4}#2\dots#2#1_{#3}}

\global\long\def\lisupwout#1#2#3#4{#1^{1}#2\dots#2\widehat{#1^{#4}}#2\dots#2#1^{#3}}

\global\long\def\lisubwoutc#1#2#3#4{#1_{1}#2\cdots#2\widehat{#1}_{#4}#2\cdots#2#1_{#3}}

\global\long\def\twp#1#2#3{\dee#1^{#2}\wedge\dee#1^{#3}}

\global\long\def\thp#1#2#3#4{\dee#1^{#2}\wedge\dee#1^{#3}\wedge\dee#1^{#4}}

\global\long\def\fop#1#2#3#4#5{\dee#1^{#2}\wedge\dee#1^{#3}\wedge\dee#1^{#4}\wedge\dee#1^{#5}}

\global\long\def\idots#1{#1\dots#1}
\global\long\def\icdots#1{#1\cdots#1}

\global\long\def\norm#1{\|#1\|}

\global\long\def\nonh{\heartsuit}

\global\long\def\nhn#1{\norm{#1}^{\nonh}}

\global\long\def\trps{^{{\scriptscriptstyle \textsf{T}}}}

\global\long\def\testfuns{\mathcal{D}}

\global\long\def\ntil#1{\tilde{#1}{}}

\global\long\def\eucl{E}

\global\long\def\mind{\alpha}
\global\long\def\vb{W}
\global\long\def\vbp{\pi}

\global\long\def\man{\mathcal{M}}
\global\long\def\odman{\mathcal{N}}
\global\long\def\subman{\mathcal{A}}

\global\long\def\vbt{\mathcal{E}}
\global\long\def\fib{\mathbf{V}}
\global\long\def\vbts{W}
\global\long\def\avb{U}

\global\long\def\chart{\varphi}
\global\long\def\vbchart{\Phi}

\global\long\def\jetb#1{J^{#1}}
\global\long\def\jet#1{j^{1}(#1)}
\global\long\def\tjet{\tilde{\jmath}}

\global\long\def\Jet#1{J^{1}(#1)}

\global\long\def\jetm#1{j_{#1}}

\global\long\def\coj{\mathfrak{d}}

\global\long\def\alt{\mathfrak{A}}

\global\long\def\pou{\eta}

\global\long\def\ext{{\textstyle \bigwedge}}
\global\long\def\forms{\Omega}

\global\long\def\dotwedge{\dot{\mbox{\ensuremath{\wedge}}}}

\global\long\def\vel{\theta}

\global\long\def\Jac{\mathcal{J}}

\global\long\def\contr{\raisebox{0.4pt}{\mbox{\ensuremath{\lrcorner}}}\,}
\global\long\def\fcontr{\raisebox{0.4pt}{\mbox{\ensuremath{\llcorner}}}\,}

\global\long\def\lie{\mathcal{L}}

\global\long\def\ssym#1#2{\ext^{#1}T^{*}#2}

\global\long\def\sh{^{\sharp}}

\global\long\def\spc{\mathcal{S}}
\global\long\def\sptm{\mathcal{E}}
\global\long\def\evnt{e}
\global\long\def\frame{\Phi}

\global\long\def\timeman{\mathcal{T}}
\global\long\def\zman{t}
\global\long\def\dims{n}
\global\long\def\m{\dims-1}
\global\long\def\dimw{m}

\global\long\def\wc{z}

\global\long\def\fourv#1{\mbox{\ensuremath{\mathfrak{#1}}}}

\global\long\def\pbform#1{\utilde{#1}}
\global\long\def\util#1{\raisebox{-5pt}{\ensuremath{{\scriptscriptstyle \sim}}}\!\!\!#1}

\global\long\def\utilJ{\util J}

\global\long\def\utilRho{\util{\rho}}

\global\long\def\body{B}
\global\long\def\man{\mathcal{M}}
\global\long\def\var{\mathcal{V}}
\global\long\def\base{\mathcal{X}}
\global\long\def\fb{\mathcal{Y}}
\global\long\def\dimb{n}
\global\long\def\dimf{m}

\global\long\def\bdry{\partial}

\global\long\def\gO{\varOmega}

\global\long\def\reg{\mathcal{R}}
\global\long\def\bdrr{\bdry\reg}

\global\long\def\bdom{\bdry\gO}

\global\long\def\bndo{\partial\gO}

\global\long\def\pis{x}
\global\long\def\xo{\pis_{0}}

\global\long\def\pib{X}

\global\long\def\pbndo{\Gamma}
\global\long\def\bndoo{\pbndo_{0}}
 \global\long\def\bndot{\pbndo_{t}}

\global\long\def\cloo{\cl{\gO}}

\global\long\def\nor{\mathbf{n}}
\global\long\def\Nor{\mathbf{N}}

\global\long\def\dA{\,\dee A}

\global\long\def\dV{\,\dee V}

\global\long\def\eps{\varepsilon}

\global\long\def\vs{\mathbf{W}}
\global\long\def\avs{\mathbf{V}}
\global\long\def\affsp{\mathbf{A}}
\global\long\def\pt{p}

\global\long\def\vbase{e}
\global\long\def\sbase{\mathbf{e}}
\global\long\def\msbase{\mathfrak{e}}
\global\long\def\vect{v}
\global\long\def\dbase{\sbase}

\global\long\def\vf{w}

\global\long\def\avf{u}

\global\long\def\stn{\varepsilon}

\global\long\def\rig{r}

\global\long\def\rigs{\mathcal{R}}

\global\long\def\qrigs{\!/\!\rigs}

\global\long\def\qd{\!/\,\!\kernel\diffop}

\global\long\def\dis{\chi}
\global\long\def\conf{\kappa}
\global\long\def\csp{\mathcal{Q}}

\global\long\def\embds{\textrm{Emb}}

\global\long\def\lc{A}

\global\long\def\lv{\dot{A}}
\global\long\def\alv{\dot{B}}

\global\long\def\fc{F}

\global\long\def\st{\sigma}

\global\long\def\bfc{\mathbf{b}}

\global\long\def\sfc{\mathbf{t}}

\global\long\def\stm{\varsigma}
\global\long\def\std{S}
\global\long\def\tst{\sigma}

\global\long\def\nhs{P}
\global\long\def\nhsa{P}
\global\long\def\nhsb{\underline{P}}

\global\long\def\soc{Z}

\global\long\def\tran{\mathrm{tr}}

\global\long\def\slf{R}

\global\long\def\sts{\varSigma}
\global\long\def\spstd{\mathfrak{S}}
\global\long\def\sptst{\mathfrak{T}}
\global\long\def\spnhs{\mathcal{P}}
\global\long\def\Ljj{\L{J^{1}(J^{k-1}\vb),\ext^{n}T^{*}\base}}

\global\long\def\spsb{\text{\Large\ensuremath{\Delta}}}

\global\long\def\ened{\mathfrak{w}}
\global\long\def\energy{\mathfrak{W}}

\global\long\def\ebdfc{T}
\global\long\def\optimum{\st^{\textrm{opt}}}
\global\long\def\scf{K}

\global\long\def\pform{\varsigma}
\global\long\def\vform{\beta}
\global\long\def\sform{\tau}
\global\long\def\flow{J}
\global\long\def\n{\m}
\global\long\def\cmap{\mathfrak{t}}
\global\long\def\vcmap{\varSigma}

\global\long\def\mvec{\mathfrak{v}}
\global\long\def\mveco#1{\mathfrak{#1}}
\global\long\def\smbase{\mathfrak{e}}
\global\long\def\spx{\simp}

\global\long\def\hp{H}
\global\long\def\ohp{h}

\global\long\def\hps{G_{\dims-1}(T\spc)}
\global\long\def\ohps{G_{\dims-1}^{\perp}(T\spc)}
\global\long\def\hpsx{G_{\dims-1}(\tspc)}
\global\long\def\ohpsx{G_{\dims-1}^{\perp}(\tspc)}

\global\long\def\fbun{F}

\global\long\def\flowm{\Phi}

\global\long\def\tgb{T\spc}
\global\long\def\ctgb{T^{*}\spc}
\global\long\def\tspc{T_{\pis}\spc}
\global\long\def\dspc{T_{\pis}^{*}\spc}

\global\long\def\fflow{\fourv J}
\global\long\def\fvform{\mathfrak{b}}
\global\long\def\fsform{\mathfrak{t}}
\global\long\def\fpform{\mathfrak{s}}

\global\long\def\maxw{\mathfrak{g}}
\global\long\def\frdy{\mathfrak{f}}
\global\long\def\ptnl{A}

\global\long\def\sobp#1#2{W_{#2}^{#1}}

\global\long\def\inner#1#2{\left\langle #1,#2\right\rangle }

\global\long\def\fields{\sobp pk(\vb)}

\global\long\def\bodyfields{\sobp p{k_{\partial}}(\vb)}

\global\long\def\forces{\sobp pk(\vb)^{*}}

\global\long\def\bfields{\sobp p{k_{\partial}}(\vb\resto{\bndo})}

\global\long\def\loadp{(\sfc,\bfc)}

\global\long\def\strains{\lp p(\jetb k(\vb))}

\global\long\def\stresses{\lp{p'}(\jetb k(\vb)^{*})}

\global\long\def\diffop{D}

\global\long\def\strainm{E}

\global\long\def\incomps{\vbts_{\yieldf}}

\global\long\def\devs{L^{p'}(\eta_{1}^{*})}

\global\long\def\incompsns{L^{p}(\eta_{1})}

\global\long\def\testf{\mathcal{D}}
\global\long\def\dists{\mathcal{D}'}

\global\long\def\codiv{\boldsymbol{\partial}}

\global\long\def\currof#1{\tilde{#1}}

\global\long\def\chn{c}
\global\long\def\chnsp{\mathbf{F}}

\global\long\def\current{T}
\global\long\def\curr{R}

\global\long\def\gdiv{\bdry\textrm{iv\,}}

\global\long\def\prop{P}

\global\long\def\aprop{Q}

\global\long\def\flux{\omega}
\global\long\def\aflux{S}

\global\long\def\fform{\tau}

\global\long\def\dimn{n}

\global\long\def\sdim{{\dimn-1}}

\global\long\def\contrf{{\scriptstyle \smallfrown}}

\global\long\def\prodf{{\scriptstyle \smallsmile}}

\global\long\def\ptnl{\varphi}

\global\long\def\form{\omega}

\global\long\def\dens{\rho}

\global\long\def\simp{s}
\global\long\def\ssimp{\Delta}
\global\long\def\cpx{K}

\global\long\def\cell{C}

\global\long\def\chain{B}

\global\long\def\ach{A}

\global\long\def\coch{X}

\global\long\def\scale{s}

\global\long\def\fnorm#1{\norm{#1}^{\flat}}

\global\long\def\chains{\mathcal{A}}

\global\long\def\ivs{\boldsymbol{U}}

\global\long\def\mvs{\boldsymbol{V}}

\global\long\def\cvs{\boldsymbol{W}}

\global\long\def\cee#1{C^{#1}}

\global\long\def\lone{L^{1}}

\global\long\def\linf{L^{\infty}}

\global\long\def\lp#1{L^{#1}}

\global\long\def\ofbdo{(\bndo)}

\global\long\def\ofclo{(\cloo)}

\global\long\def\vono{(\gO,\rthree)}

\global\long\def\vonbdo{(\bndo,\rthree)}
\global\long\def\vonbdoo{(\bndoo,\rthree)}
\global\long\def\vonbdot{(\bndot,\rthree)}

\global\long\def\vonclo{(\cl{\gO},\rthree)}

\global\long\def\strono{(\gO,\reals^{6})}

\global\long\def\sob{W_{1}^{1}}

\global\long\def\sobb{\sob(\gO,\rthree)}

\global\long\def\lob{\lone(\gO,\rthree)}

\global\long\def\lib{\linf(\gO,\reals^{12})}

\global\long\def\ofO{(\gO)}

\global\long\def\oneo{{1,\gO}}
\global\long\def\onebdo{{1,\bndo}}
\global\long\def\info{{\infty,\gO}}

\global\long\def\infclo{{\infty,\cloo}}

\global\long\def\infbdo{{\infty,\bndo}}

\global\long\def\ld{LD}

\global\long\def\ldo{\ld\ofO}
\global\long\def\ldoo{\ldo_{0}}

\global\long\def\trace{\gamma}

\global\long\def\pr{\proj_{\rigs}}

\global\long\def\pq{\proj}

\global\long\def\qr{\,/\,\reals}

\global\long\def\aro{S_{1}}
\global\long\def\art{S_{2}}

\global\long\def\mo{m_{1}}
\global\long\def\mt{m_{2}}

\global\long\def\yieldc{B}

\global\long\def\yieldf{Y}

\global\long\def\trpr{\pi_{P}}

\global\long\def\devpr{\pi_{\devsp}}

\global\long\def\prsp{P}

\global\long\def\devsp{D}

\global\long\def\ynorm#1{\|#1\|_{\yieldf}}

\global\long\def\colls{\Psi}

\global\long\def\semib{\mathrm{SB}}

\global\long\def\tm#1{\overrightarrow{#1}}
\global\long\def\tmm#1{\underrightarrow{\overrightarrow{#1}}}

\global\long\def\itm#1{\overleftarrow{#1}}
\global\long\def\itmm#1{\underleftarrow{\overleftarrow{#1}}}

\global\long\def\ptrac{\mathcal{P}}

\title[$k$-Jet Hyper-Stresses]{Hyper-Stresses in $k$-Jet Field Theories}

\author{Reuven Segev and J\k{e}drzej \'{S}niatycki}

\address{Reuven Segev\\
Department of Mechanical Engineering\\
Ben-Gurion University of the Negev\\
Beer-Sheva, Israel\\
rsegev@bgu.ac.il}

\address{J\k{e}drzej \'Sniatycki\\
Departments of Mathematics and Statistics, University of Calgary, Calgary, Alberta, and  University of Victoria, B.C., Canada,   sniatycki@gmail.com\medskip \medskip \\}


\thanks{\today}

\begin{abstract}
For high-order continuum mechanics and classical field theories configurations
are modeled as sections of general fiber bundles and generalized velocities
are modeled as variations thereof. Smooth stress fields are considered
and it is shown that three distinct mathematical stress objects play
the roles of the traditional stress tensor of continuum mechanics
in Euclidean spaces. These objects are referred to as the variational
hyper-stress, the traction hyper-stress and the non-holonomic stress.
The properties of these three stress objects and the relations between
them are studied.

\end{abstract}

\maketitle

\section{Introduction}

This manuscript is concerned with hyper-stresses for theories formulated
on general fiber bundles over differentiable manifolds. Although the
study of higher order continuum mechanics entered the focus of attention
early in the second half of the 20th century (\eg, , \cite{toupin_elastic_1962,toupin_theories_1964,mindlin_micro-structure_1964,mindlin_second_1965}
it is still the subject of current research from both theoretical
(\eg, \cite{noll_edge_1990,dellisola_postulations_2015,dellisola_how_2012,podio-guidugli_cauchys_2015,fosdick_generalized_2016})
and practical (\eg, \cite{aifantis_role_1992,ru_simple_1993,askes_gradient_2011,bertram_compendium_2017})
interests.

In parallel, analogous questions engaged the mathematical physics
literature considering classical field theories (\eg, \cite{kijowski_symplectic_1979,leon_generalized_1985,binz_geometry_1988,gotay_momentum_2003,giachetta_advanced_2009}),
although the variational approach has been much more pronounced. While
in continuum mechanics fields are usually defined on a three-dimensional
body manifold or the three-dimensional physical space, classical field
theories are mostly formulated as sections of fiber bundles over space-time
manifold.

Following Walter Noll's \cite{noll_foundations_1959}, it is by now
accepted that the body object of continuum mechanics should be modeled
as a manifold devoid of a metric or a connection. Furthermore, fields
defined on the body, such as various order parameters or internal
degrees of freedom, cannot be always expected to be valued in Euclidean
spaces. Thus, formulations in the setting of general manifolds are
relevant in continuum mechanics as much as they are for physical field
theories.

While traditionally, $k$-th order hyper-stresses emerge in continuum
mechanics as derivatives of a strain energy density with respect to
derivatives of order $k$ of the configuration, in \cite{segev_forces_1986,segev_consistency_1991}
their existence follow from a representation theorem for forces which
are conjugate to the space of $C^{k}$-velocity fields. This setting
has the advantage that it applies to the geometry of general manifolds
and to hyper-stresses that may be as irregular as Borel measures.
It is noted that for a theory of order $k$ on a differentiable manifold,
the derivatives of a particular order $0<r\le k$, do not give rise
to an invariant geometric object. Rather, one has to use the $k$-jet
of the generalized velocity field combining all derivatives of order
less or equal to $k$ into a single invariant object. (See \cite{saunders_geometry_1989}
as a general reference on jet bundles.)

Among the peculiarities of continuum mechanics on manifolds, even
in the smooth case considered in this paper and $k=1$, we emphasize
the distinction between variational stresses and traction stresses
\cite{segev_notes_2013}. For classical continuum mechanics in a Euclidean
space, the same mathematical object determines the traction on subbodies\textemdash the
traction stress\textemdash and acts on velocity gradients to produce
power\textemdash the variational stress, in our terminology. Yet,
for the general geometric setting, two distinct mathematical objects
play these roles. While the variational stress determines the traction
stress, there is a class of traction stresses for each variational
stress field.

As we show below, it turns out that for high order theories, a variational
hyper-stress does not determine a unique traction hyper-stress. In
fact, a traction hyper-stress field together with the body hyper-force
encode more information than a variational hyper-stress field. Now,
a third object is needed, to which we refer as the non-holonomic stress
field. While the non-holonomic stress has been introduced in \cite{segev_geometric_2017}
for the case $k=2$ as a generalization of the variational stress
enabling one to apply integral transformations, here we show that
non-holonomic stresses are the objects that determine traction hyper-stresses
as defined in \cite{segev_jets_2017}.

In order to introduce the terminology and exhibit the various roles
played by these three distinct mathematical objects in the simplest
terms, consider the case where the body $\reg$ is identified with
its current configuration in $\reals^{3}$. Let $\vf^{j}$ be the
components of a velocity field and let $\vf_{,\mie ik}^{j}$ be its
$k$-th partial derivatives. The components of the variational stress
of highest order $\std_{j}^{\mie ik}$ act on the derivatives of the
velocity in the form
\begin{equation}
I=\int_{\reg}\std_{j}^{\mie ik}\vf_{,\mie ik}^{j}\dV.
\end{equation}
Since the partial derivatives are symmetric relative to permutations
of the indices $\mie ik$, the components of the variational stress
satisfy the same symmetry conditions. This symmetry condition for
the components of the variational stress emerges naturally in the
hyper-elastic case where the components of the variational stress
are given in terms of a potential function $\vph$ and the symmetric
components of the $k$-th partial derivatives of the configuration
$\conf_{,\mie ik}^{j}$ by
\begin{equation}
\std_{j}^{\mie ik}=\dby{\vph}{\conf_{,\mie ik}^{j}}.
\end{equation}

Returning to the expression for the action of the variational stress,
one has
\begin{equation}
\begin{split}I & =\int_{\reg}\left[\bigp{\std_{j}^{\mie ik}\vf_{,\mie i{k-1}}^{j}}_{,i_{k}}-\std_{j,i_{k}}^{\mie ik}\vf_{,\mie i{k-1}}^{j}\right]\dV,\\
 & =\int_{\bdry\reg}\std_{j}^{\mie ik}n_{i_{k}}\vf_{,\mie i{k-1}}^{j}\dA-\int_{\reg}\std_{j,i_{k}}^{\mie ik}\vf_{,\mie i{k-1}}^{j}\dV,
\end{split}
\end{equation}
where $n$ is the unit normal to the boundary. We will refer to a
tensor such as $\sfc_{j}^{\mie i{k-1}}=\std_{j}^{\mie ik}n_{i_{k}}$
as a hyper-traction and to $\bfc_{j}^{\mie i{k-1}}=\std_{j,i_{k}}^{\mie ik}$
as a body hyper-force. Thus, in the expression
\begin{equation}
\sfc_{j}^{\mie i{k-1}}=\std_{j}^{\mie ik}n_{i_{k}},
\end{equation}
the variational hyper-stress $\std$ acts as a traction hyper-stress
by determining the hyper-traction on the boundary. That is, a traction
hyper-stress $\tst_{j}^{\mie ik}$ determines hyper-tractions by 
\begin{equation}
\sfc_{j}^{\mie i{k-1}}=\tst_{j}^{\mie ik}n_{i_{k}}.
\end{equation}
Nevertheless, if the role of the traction hyper-stress is exhibited
in the last equation, there is no a-priori reason to assume that $\tst_{j}^{\mie ik}$
is symmetric relative to permutations of all indices $\mie ik$ and
not merely $\mie i{k-1}$. In other words a traction hyper-stress
need not be determined by a variational hyper-stress, but rather,
by a more general objects to which we refer as the non-holonomic stress.

One could postulate that traction hyper-stresses are indeed symmetric
relative to all the indices $\mie ik$. However, it turns out that
this condition cannot be formulated invariantly in the general geometric
case. Indeed, on general fiber bundles, a traction hyper-stress field
together with a body hyper-force determine a unique non-holonomic
hyper-stress field. A non-holonomic hyper-stress field does not act
on $k$-jets of velocity fields but rather, on the first jets of sections
of the $(k-1)$-jet bundle. Such sections of the $(k-1)$-jet bundle
need not be holonomic\textemdash compatible\textemdash hence the terminology.
It follows that non-holonomic hyper-stress fields obey less restrictive
symmetry conditions in comparison with variational stresses. Conversely,
a non-holonomic hyper-stress field determines a unique traction hyper-stress
field. Furthermore, a non-holonomic hyper-stress determines a unique
variational stress that may be used to compute the force on each subbody.
However, this variational hyper-stress does not encode enough information
as to determine the traction hyper-stress on each subbody. We find
it somewhat intriguing that for higher-order continuum mechanics\textemdash originally
based on the variational approach\textemdash it turns out that the
generalization of the Cauchy construction plays such a crucial role.

Section \ref{sec:Notation-and-Preliminaries} introduces the notation
and some of the terminology used in the sequel, mainly, that corresponding
to symmetric tensors. Section \ref{sec:Continuum-Mechanics-Jets}
reviews the fundamentals of continuum mechanics and field theories
on $k$-jet bundles of fiber bundles. In particular, variational hyper-stress
densities are introduced. Section \ref{sec:Hyper-Stresses,-the-Cauchy}
outlines higher-order theories and introduces traction hyper-stresses
using a framework which generalizes the classical Cauchy approach.
Section \ref{sec:Non-Holonomic-Variational-Stress} considers non-holonomic
stresses and their relations with traction stresses. Next, the variational
stress induced by a non-holonomic stress is presented in Section \ref{sec:Induced-Variational-Stresses}.
\emph{ }Finally, in Section \ref{sec:Spaces-of-Hyper-Stresses},
we describe the geometric setting in which elastic constitutive relations
are formulated.

\section{Notation and Preliminaries\label{sec:Notation-and-Preliminaries}}

Multi-index notation, as used in this manuscript, will help making
the algebraic expressions below more compact. In this section we introduce
the notation adopted here, with particular application to symmetric
tensors. (For additional details, see \cite{segev_jets_2017}.)

\subsection{Multi-indices}

A collection of indices $i_{1}\cdots i_{k}$, $i_{r}=\oneto n$ will
be represented as a multi-index $I$ and we will write $\abs I=k$,
the length of the multi-index. Multi-indices will be denoted by upper-case
roman letters and the associated indices will be denoted by the corresponding
lower case letters. In what follows, we will use the summation convention
for repeated indices as well as repeated multi-indices. Whenever the
syntax is violated, \eg, when a multi-index appears more than twice
in a term, it is understood that summation does not apply.

A multi-index $I$ induces a sequence $(\lisub I,n)$ in which $I_{r}$
is the number of times the index $r$ appears in the sequence $\mie ik$.
Thus, $\abs I=\sum_{r}I_{r}$. Multi-indices may be concatenated naturally
such that $\abs{IJ}=\abs I+\abs J$. Additional convenient notation
is introduced by 
\begin{equation}
I!:=I_{1}!\cdots I_{n}!,\qquad\gd_{J}^{I}:=\gd_{j_{1}}^{i_{1}}\cdots\gd_{j_{k}}^{i_{k}}.
\end{equation}

Greek letters will be used for strictly increasing multi-indices used
in the representation of alternating tensors and forms, \eg, 
\begin{equation}
\go=\go_{\gl}\dee x^{\gl}:=\go_{\mie{\gl}{\abs{\gl}}}\lisuppc{\dee x}{\gl_{1}}{\wedge}{\gl_{\abs{\gl}}}.
\end{equation}
To simplify the notation, the local volume element induced by a coordinate
system will be denoted by $\dee x$, \ie, 
\begin{equation}
\dee x=\lisuppc{\dee x}1{\wedge}n.
\end{equation}

\subsection{Permutations}

We denote by $\sperm_{l}$ the group of permutations of $(\oneto l)$.
For a multi-index $I=\mie il$ and a permutation $p\in\sperm_{l}$,
\begin{equation}
p(I)=I\comp p=i_{\perm(1)}\cdots i_{p(l)}.
\end{equation}
An $l$-permutation $\perm$ acts on an $l$-dimensional array by
$\perm(T)_{I}=T_{\perm(I)}$.

We observe that for a multi-index $I$, of all $\abs I!$ permutations
$\perm(I)$, there are $I!$ permutations that leave $I$ invariant.
Thus, there are $\abs I!/I!$ elements in the collection, $\sperm_{I}$,
containing permutations that give distinct multi-indices $J=\perm(I)$. 

We will also use the notation
\begin{equation}
\abs{\eps}_{J}^{I}=\begin{cases}
1, & \text{if }I\text{ may be obtained by a permutation of \ensuremath{J,}}\\
0, & \text{otherwise.}
\end{cases}
\end{equation}

\subsection{Symmetric tensors}

We will view an $l$-tensor as an $l$-multilinear mapping, an element
of $L^{l}(\avs,\vs)\simeq\bigp{\tp^{l}\avs^{*}}\otimes\vs$ for two
vector spaces $\avs$ and $\vs$. In order to simplify the notation,
for the rest of this section, we will use $\vs=\reals$. For other
finite dimensional vector spaces, the extension is straightforward.
A permutation $\perm$ acts on an $l$-tensor $T$ by
\begin{equation}
\perm(T)(\lisubb v1,l)=T(\lisubb v{\perm(1)},{\perm(l)})
\end{equation}
which implies that the array of $\perm(T)$ is indeed $\perm(T)_{I}=T_{\perm(I)}$.
A symmetric tensor $T$ satisfies the condition
\begin{equation}
\perm(T)=T,\qquad\text{or equivalently,}\qquad T_{p(I)}=T_{I},
\end{equation}
for all permutations $\perm\in\sperm_{l}$. A symmetric array is uniquely
determined by the elements of the form 
\begin{equation}
T_{\ordr I}=T_{\mii},
\end{equation}
where we use the convention that multi-indices denoted by bold faced
characters are non-decreasing, that is $\lisubbc i1{\le}l$, and alternatively,
when a multi-index appears inside angle brackets, it is implied that
it is permuted so that it is non-decreasing. When the summation convention
should be applied to such multi-indices, it is understood that the
summation is carried out only for non-decreasing occurrences.

For an $n$-dimensional vector space $\avs$, the subspace of symmetric
tensors, $L_{S}^{l}(\avs,\reals)$, has the dimension (see \cite{lewis_notes_nodate}),
\begin{equation}
\dimension L_{S}^{l}(\avs,\reals)=\frac{(n+l-1)!}{(n-1)!l!}.
\end{equation}

\subsection{Symmetrization}

We will use the symmetrization operator
\begin{equation}
\smtr:L^{l}(\avs,\vs)\tto L_{S}^{l}(\avs,\reals),\qquad\text{whereby,}\qquad\smtr(T)=\frac{1}{l!}\sum_{\perm\in\sperm_{l}}\perm(T),
\end{equation}
which implies that the array of the symmetrized tensor is the symmetrized
array, \ie, 
\begin{equation}
T_{\symm I}:=\smtr(T)_{I}=\frac{1}{l!}\sum_{\perm\in\sperm_{l}}T_{\perm(I)}.
\end{equation}
The symmetrized tensor product is $T\odot R=\smtr(T\otimes R)$. Using
the notation
\begin{equation}
\sbase^{J}:=\tp^{J}\sbase^{J}=\lisuppc{\sbase}{j_{1}}{\otimes}{j_{l}},\qquad\sbase^{\symm J}:=\stp^{J}\sbase^{J}=\lisuppc{\sbase}{j_{1}}{\odot}{j_{l}},
\end{equation}
a basis for the space of symmetric tensors may be formed by the elements$\{\sbase^{\symm{\mii}}\}$
(where we observe the $\mii$ multi-indices are non-decreasing). Thus,
one has the identification
\begin{equation}
L_{S}^{l}(\avs,\reals)\simeq\stp^{l}\avs^{*}.
\end{equation}

Consider the inclusion 
\begin{equation}
\incl_{S}:\stp^{l}\avs^{*}\tto\tp^{l}\avs^{*}.
\end{equation}
For a symmetric tensor $T$, one has
\begin{equation}
\begin{split}\incl_{S}(T) & =T=T_{J}\sbase^{J}=T_{J}\sbase^{\symm J},\\
 & =\sum_{\mii}\sum_{\perm\in\sperm_{\mii}}T_{\perm(\mii)}\sbase^{\symm{\perm(\mii)}},\\
 & =\sum_{\mii}\frac{\abs{\mii}!}{\mii!}T_{\mii}\sbase^{\symm{\mii}}\qquad\text{(no sum).}
\end{split}
\end{equation}
If we want the components of $T$ relative to the basis in $\odot^{l}\avs^{*}$
to be equal to the corresponding components of the array of $\incl_{S}(T)=T$,
we should use the modified basis 
\begin{equation}
\itmm{\sbase}^{\symm{\mii}}:=\frac{\abs{\mii}!}{\mii!}\sbase^{\symm{\mii}},\qquad\text{so that}\qquad T=T_{\mii}\itmm{\sbase}^{\symm{\mii}}.
\end{equation}
Similarly, one can easily show that for the primal basis $\{\sbase_{i}\}$,
\begin{equation}
\itmm{\sbase}^{\symm{\mii}}(\sbase_{\symm{\mi J}})=\gd_{\mi J}^{\mii},\qquad\sbase^{\symm{\mii}}(\itmm{\sbase}_{\symm{\mi J}})=\gd_{\mi J}^{\mii},
\end{equation}
and so $\{\itmm{\sbase}^{\symm{\mii}}\}$ may serve as the dual basis
of $\{\sbase_{\symm{\mi J}}\}$ and vice versa. 

For a basis of the tangent space to a manifold induced by the coordinates
$(x^{i})$, we will use the basis consisting of the elements
\begin{equation}
\bdry_{\mii}:=\bdry_{i_{1}}\odot\cdots\odot\bdry_{i_{l}},\qquad\bdry_{i}:=\frac{\bdry}{\bdry x^{i}},
\end{equation}
for the space of symmetric tensors, while the dual basis consists
of the elements
\begin{equation}
\itmm{\dee x}^{\mii}:=\frac{\abs{\mii}!}{\mii!}\lisuppc{\dee x}{i_{1}}{\odot}{i_{l}}.\label{eq:Basis_Elem_Jets_DEF}
\end{equation}
It is noted that in both cases, only non-decreasing indices are used.

It follows that $\odot^{l}\avs^{*}$ may be identified with $\bigp{\odot^{l}\avs}^{*}$.
In analogy with the above, for $R=R^{\mii}\sbase_{\symm{\mii}}\in\odot^{l}\avs$,
$T=T_{\mi J}\itmm{\sbase}^{\mi J}$, $T(R)=T_{\mii}R^{\mii}$, and
\begin{equation}
T_{J}R^{J}=\sum_{\mii}\frac{\abs{\mii}!}{\mii!}T_{\mii}R^{\mii}\qquad\text{(no sum).}
\end{equation}
So setting, 
\begin{equation}
\itmm R^{\mii}=\frac{\abs{\mii}!}{\mii!}R^{\mii},\qquad\tmm R^{J}=\frac{J!}{\abs J!}R^{J},\label{eq:TMM-defined}
\end{equation}
one has 
\begin{equation}
T_{J}R^{J}=T_{\mii}\itmm R^{\mii},\qquad T_{J}\tmm R^{J}=T_{\mii}R^{\mii}=T(R).\label{eq:TMM-used}
\end{equation}
The last expression may be interpreted as a statement that the components
$\tmm R^{J}$ represent the inclusion of the symmetric tensor $R$
in the space of all tensors.

We will refer to elements of $\bigp{\odot^{l-1}\avs^{*}}\otimes\avs^{*}$
as \emph{almost symmetric} tensors. It can be shown \cite{lewis_notes_nodate}
that for an almost symmetric tensor $T$,
\begin{equation}
\smtr(T)=\frac{1}{l}\sum_{\perm\in\tilde{\sperm}_{l}}\perm(T),
\end{equation}
where $\tilde{\sperm}_{l}$ contains permutations $\perm$ such that
$p(l)=\oneto l$, and $p(1)<\cdots<p(l-1)$, that is, $\perm$ switches
$l$ with another number and then orders the first $l-1$ numbers.

For a symmetric array $\vf_{I}$ and an array $T^{J}$, $\abs I=\abs J$,
we will encounter below expressions such as 
\begin{equation}
T^{\mii}w_{\mii}=T^{\symm{\mii}}w_{\mii},\quad\text{and}\quad T^{\mi Jj}w_{\mi Jj}=T^{\symm{\mi Jj}}w_{\mi Jj}.
\end{equation}
We wish to find the relation between these two expressions. Note that
each pair $\mi Jj$, determines by ordering a unique $\mii$ represented
by $(J_{1},\dots,J_{j}+1,\dots,J_{n})$. 

Consider a particular $\mii$ represented by $(\lisub I,n)$. Let
\begin{equation}
\{\mii\}=\{r\mid I_{r}>0\},\qquad c(\mii)=\mathrm{cardinality}\{\mii\},
\end{equation}
\ie, $c(\mii)$ is the number of indices that appear in $\mii$ one
time or more.  Since a pair $\mi Jj$ determines a unique $\mii$,
we may also write $c(\mi Jj)$. Let $j\in\{\mii\}$, then, 
\begin{equation}
\mi Jj,\qquad\text{with \qquad}\mi J=(I_{1},\dots,I_{j}-1,\dots,I_{n})
\end{equation}
will give $\mii$ upon rearranging. Thus, for each $\mii$ there are
exactly $c(\mii)$ pairs $\mi Jj$ that may be rearranged to give
$\mii$. It follows that
\begin{equation}
T^{\mi Jj}w_{\mi Jj}=\sum_{j\in\{\mii\}}T^{\mi Jj}w_{\mi Jj}=c(\mii)T^{\mii}w_{\mii}\qquad\text{(no sum on }\mii).
\end{equation}
Define the arrays $\tm T^{\mi Jj}$ and $\itm T^{\mii}$ by
\begin{equation}
\tm T^{\mi Jj}=\frac{1}{c(\mi Jj)}T^{\mi Jj},\qquad\text{and}\qquad\itm T^{\mii}=c(\mii)T^{\mii}.\label{eq:Def-ModT}
\end{equation}
It is concluded that
\begin{equation}
\tm T^{\mi Jj}w_{\mi Jj}=T^{\mi I}w_{\mii},\qquad\text{and}\qquad T^{\mi Jj}w_{\mi Jj}=\itm T^{\mii}w_{\mii}.\label{eq:ModT-Prop}
\end{equation}

\subsection{Jets\label{subsec:Jets}}

For a fiber bundle $\xi:\fb\to\base$, we will denote by $\xi^{k}:J^{k}(\base,\fb)\to\base$
the corresponding $k$-jet bundle of sections of $\xi$. When no ambiguity
may occur, we will often use the simpler notation $\xi^{k}:J^{k}\fb\to\base$.
One has the additional natural projections $\xi_{l}^{k}:J^{k}(\base,\fb)\to J^{l}(\base,\fb)$,
$l<k$, and in particular $\xi_{0}^{k}:J^{k}(\base,\fb)\to\fb$, \cite{saunders_geometry_1989}.

Let $\eta:\vb\to\base$ be a vector bundle and let $\eta^{k}:J^{k}\vb\to\base$
be the corresponding $k$-jet bundle. For $l<k$, the $l$-vertical
subbundle of the jet bundle is defined by
\begin{equation}
V^{l}J^{k}\vb=\kernel\eta_{l}^{k}=\{\lv\in J^{k}\vb\mid\eta_{l}^{k}(\lv)=0\}.
\end{equation}
Evidently, $V^{l}J^{k}\vb$ is a vector subbundle of the jet bundle
and we denote the inclusion by 
\begin{equation}
\incl_{V_{k}^{l}}:V^{l}J^{k}\vb\tto J^{k}\vb.\label{eq:Vertical_Subb}
\end{equation}

\subsection{Pullback of forms}

Here we use ``$\sharp$'' to indicate the pullback of a form by duality,
rather than using ``$*$''. The latter is used to indicate the pullback
of a form viewed as a section of the pullback bundle. 

\section{Continuum Mechanics and Field Theories on $k$-Jet Bundles\label{sec:Continuum-Mechanics-Jets}}

The fundamental object we consider is a fiber bundle
\begin{equation}
\xi:\fb\tto\base
\end{equation}
where $\base$ is an $\dimb$-dimensional orientable manifold and
the typical fiber is an $\dimf$-dimensional manifold. For the sake
of simplicity, it is assumed that $\base$ is compact. This fiber
bundle may have the following interpretations. 

In Lagrangian continuum mechanics, it is usually assumed that the
total space $\fb$ is trivial, \ie, 
\begin{equation}
\fb=\base\times\spc
\end{equation}
where the manifold $\base$ is interpreted as the body manifold and
the manifold $\spc$ is interpreted as the space manifold. In this
case, a section $\conf:\base\to\fb$ may be identified with a mapping
$\base\to\spc$. Such a mapping is interpreted as a configuration
of the body in the physical space.

In Eulerian continuum mechanics, $\base$ is interpreted as the physical
space and the fibers of $\fb$ are interpreted as the possible values
that fields over the space may have. For classical field theories,
$\base$ is interpreted as space-time. (See \cite{kupferman_stress_2017}
for further motivation, references and examples.) 

\subsection{The configuration space}

A configuration of the system, or a field, is a section $\conf:\base\to\fb$.
In accordance with the paradigm of analytical mechanics, a central
role is played by the configuration space of the system, an infinite
dimensional object in the case of a field theory. We are motivated
by the case of continuum mechanics where one requires traditionally
that configurations of a body in space be embeddings, and the fact
that the subset of embeddings is open in the manifold of all $C^{k}$-mappings
between two manifolds, for $k\ge1$. (See \cite{michor_manifolds_1980,m._hirsch_differential_1976}.)
Thus, for $k\ge1$, we consider the collection of $C^{k}$-sections
of $\xi$, equipped with the $C^{k}$-topology, and define the configuration
space $\csp$ to be an open subset of the manifold of sections, \ie, 
\begin{equation}
\csp\subset C^{k}(\xi):=C^{k}(\base,\fb),\quad\text{is open.}
\end{equation}
Note that by the notation $C^{k}(\base,\fb)$ we refer only to sections
of the fiber bundle rather than all mappings $\base\to\fb$.

\subsection{Generalized velocities}

Generalized velocities (virtual velocities, virtual displacements)
at the configuration $\conf$, or variations of the configuration
$\conf$, are elements of the tangent space $T_{\conf}\csp$. Since
$\csp$ is assumed to be open, 
\begin{equation}
T_{\conf}\csp=T_{\conf}C^{k}(\base,\fb).
\end{equation}
Furthermore, let 
\begin{equation}
V\xi=V\fb:=\kernel T\xi
\end{equation}
be the vertical subbundle of $T\fb$. Then (see \cite[ p. 51]{palais_foundations_1968}),
\begin{equation}
T_{\conf}\csp=T_{\conf}C^{k}(\base,\fb)\simeq C^{k}(\base,\conf^{*}V\fb),
\end{equation}
where $\conf^{*}V\fb$ denotes the pullback of the vertical bundle
onto $\base$. Thus, every generalized velocity may be identified
uniquely with a $C^{k}$-section 
\begin{equation}
\vf:\base\tto\conf^{*}V\fb,
\end{equation}
in accordance with the traditional interpretation in continuum mechanics.
(It is noted that in \cite{palais_foundations_1968}, $V\fb$ is denoted
by $TF(\fb)$ and $\conf^{*}V\fb$ is denoted by $T_{\conf}\fb$.)

Thus, 
\begin{equation}
TC^{k}(\base,\fb)\simeq C^{k}(\base,V\fb)
\end{equation}
and a section $\vf$ represents an element of $T_{\conf}\csp$ if
\begin{equation}
\tau_{\fb}\resto{V\fb}\comp\vf=\conf.
\end{equation}
In the sequel, in order to simplify the notation, we will denote the
vector bundle
\begin{equation}
\xi\comp\tau_{\fb}\resto{V\fb}:\conf^{*}V\fb\tto\base\qquad\text{by}\qquad\pi:\vb\tto\base,\label{eq:VYisW}
\end{equation}
Using local coordinates $(x^{i},y^{\ga})$ in $\fb$, an element of
$T\fb$ is represented locally in the form $\dot{x}^{i}\partial_{i}+\dot{y}^{\ga}\partial_{\ga}$.
For an element of $V\fb$, $\dot{x}^{i}=0$, and so, an element of
$\conf^{*}V\fb=W$ is represented locally in the form $\dot{y}^{\ga}\conf^{*}\partial_{\ga}$,
which we may write also simply as $\dot{y}^{\ga}\sbase_{\ga}$, $\sbase_{\ga}:=\conf^{*}\partial_{\ga}$.
The dual bases will be denoted accordingly. However, in what follows
we often identify $\conf^{*}\partial_{\ga}$ with $\bdry_{\ga}$ in
the notation.

\subsection{Generalized forces and their representations by variational hyper-stresses\label{subsec:Forces_and_Stress_Repr}}

A generalized force at a configuration $\conf$ is an element of $T_{\conf}^{*}\csp$.
As such, a force $\fc$ at $\conf$ is a continuous linear functional
in $C^{k}(\base,\conf^{*}V\fb){}^{*}$. Henceforth, we will use the
concise notation and will write 
\begin{equation}
T_{\conf}\csp=C^{k}(\conf^{*}V\fb)=C^{k}(\vb),\qquad\fc\in C^{k}(\conf^{*}V\fb)^{*}=C^{k}(\vb)^{*}.
\end{equation}

We observe that the jet extension mapping
\begin{equation}
j^{k}:C^{k}(\vb)\tto C^{0}(J^{k}\vb),\qquad\vf\lmt j^{k}\vf,
\end{equation}
is an embedding. (It is not surjective as non-holonomic sections are
not included in the image.) It follows that its dual mapping,
\begin{equation}
j^{k*}:C^{0}(J^{k}\vb)^{*}\tto C^{k}(\vb)^{*},
\end{equation}
is surjective. Using the Hahn-Banach theorem, one concludes that for
a given force $\fc$, there is some element $\stm\in C^{0}(\base,J^{k}\vb)^{*}$
such that
\begin{equation}
\fc=j^{k*}\stm\qquad\text{or}\qquad\fc(\vf)=\stm(j^{k}\vf),\label{eq:equil-eq}
\end{equation}
for every $\vf\in C^{k}(\vb)$. Equations (\ref{eq:equil-eq}) represent
the generalizations of the equation of equilibrium and the principle
of virtual work of continuum mechanics. The fact that $j^{k}$ is
not surjective, implies that the representation is not unique. This
non-uniqueness of the stress representation of a given force is a
generalization of the well-known static indeterminacy of continuum
mechanics.

Since $\stm$ is a linear functional on the space of continuous sections
of the vector bundle $\pi^{k}:J^{k}\vb\to\base$, it may be represented
by a measure $\mu$ valued in $(J^{k}\vb)^{*}$ in the form
\begin{equation}
\stm(\lv)=\int_{\base}\lv\cdot\dee\mu
\end{equation}
for any section $\lv$ of $\pi^{k}$. 

For a chart in $\base$ with coordinates $(x^{i})$ and a local basis
$\{\sbase_{\ga}\}$ for $\vb$, a section of $\vb$ is represented
locally in the form 
\begin{equation}
\vf=\vf^{\ga}\sbase_{\ga}
\end{equation}
and its jet is represented in the form
\begin{equation}
j^{k}\vf=\vf_{,\mii}^{\ga}\sbase_{\ga}^{\mii},\qquad0\le\abs I\le k,
\end{equation}
where
\begin{equation}
\sbase_{\ga}^{\mii}:=\itmm{\dee x}^{\mii}\otimes\sbase_{\ga}.
\end{equation}
Thus, a section $\lv$ of the jet bundle is represented in the form
\begin{equation}
\lv=\lv_{\mii}^{\ga}\sbase_{\ga}^{\mii}\qquad0\le\abs I\le k.
\end{equation}
Locally, $\mu$ is represented by a collection of real valued measures
$\mu_{\ga}^{\mi I}$, $\abs{\mi I}\le k$ so that
\begin{equation}
\lv\cdot\dee\mu=\lv_{\mi I}^{\ga}\cdot\dee\mu_{\ga}^{\mi I}.
\end{equation}

The representing elements $\stm\in C^{0}(J^{k}\vb)^{*}$ are the \emph{hyper-stresses}
as implied by the principal of virtual work above. Furthermore, their
role as a means for restriction of forces results from their representation
by measures. Thus, a hyper-stress $\stm$ represented by the measure
$\mu$, induces for any $n$-dimensional submanifold $\reg$, a force
$\fc_{\reg}$ by
\begin{equation}
\fc_{\reg}(\vf)=\int_{\reg}j^{1}\vf\cdot\dee\mu.
\end{equation}
Hence, while forces cannot be restricted to submanifolds, the importance
of the representation by hyper-stresses follows from the fact that
measures may be restricted to Borel sets.

\subsection{Smooth variational hyper-stresses\label{subsec:Smooth-variational-hyper-stresse}}

In what follows, we consider a smooth configuration $\conf$ and set
$\vb=\conf^{*}V\fb$. Furthermore, the focus of this work is placed
on smooth force distributions and hyper-stresses, in other words,
forces and hyperstresses represented by smooth densities. Consider
the vector bundle of linear mappings
\begin{equation}
\L{J^{k}\vb,\ext^{n}T^{*}\base}\simeq(J^{k}\vb)^{*}\otimes_{\base}\ext^{n}T^{*}\base.
\end{equation}
A \emph{smooth variational hyper-stress field} is a smooth section
$\std$ of $\L{J^{k}\vb,\ext^{n}T^{*}\base}$. Alternatively, a variational
hyper-stress field may be viewed as a smooth vector bundle morphism
\begin{equation}
\std:J^{k}\vb\tto\ext^{n}T^{*}\base,\label{eq:VarSt_as_VB-Mor}
\end{equation}
where both bundles have $\base$ as the base manifold. A variational
hyper-stress field $\std$ induces a stress $\stm$ by
\begin{equation}
\stm(\lv)=\int_{\base}\std\cdot\lv.\label{eq:SmthSt}
\end{equation}
The integral on the right hand side makes sense as $\std\cdot\lv$
is an $n$-form on $\base$. In particular, the power that a force
$\fc$ expends for the generalized velocity $\vf$ is given by
\begin{equation}
\fc(\vf)=\stm(j^{k}\vf)=\int_{\base}\std\cdot j^{k}\vf,\label{eq:SmthStPower}
\end{equation}
and the restriction of the force, induced by the stress, to the subbody
$\reg$ is given by
\begin{equation}
\fc_{\reg}(\vf)=\int_{\reg}\std\cdot j^{k}\vf.\label{eq:RestForce}
\end{equation}

To consider the local representation of the variational hyper-stress
densities, we first note that as the dual basis to $\Bigl\{\itmm{\dee x}^{\mii}\Bigr\}$
is $\{\bdry_{\mii}\}$, the dual basis corresponding to $\{\sbase_{\ga}^{\mii}\}$
is 
\begin{equation}
\left\{ \sbase_{\mii}^{\ga}:=\partial_{\mi I}\otimes\sbase^{\ga}\right\} .
\end{equation}
 It follows that a stress density is represented locally in the form
\begin{equation}
\std_{\ga}^{\mi I}\sbase_{\mii}^{\ga}\otimes\dee x,\qquad0\le\abs{\mii}\le k,\label{eq:Rep_Var_Str}
\end{equation}
and the action of the variational stress density on a section $\lv$
of the jet bundle is given by
\begin{equation}
\std\cdot\lv=\std_{\ga}^{\mii}\lv_{\mii}^{\ga}\dee x,\qquad\text{in particular,}\qquad\std\cdot j^{k}\vf=\std_{\ga}^{\mii}\vf_{,\mii}^{\ga}\dee x.\label{eq:SmthStComp}
\end{equation}

Consider next the transformation rule for the components $\std_{\ga}^{\mii}$.
Let $x^{i'}(x^{i})$ be a transformation of coordinates in a subset
of $\base$ and $\vf^{\ga'}=\lv_{\ga}^{\ga'}(x)\vf^{\ga}$ be the
coordinate transformation in $\vb$. It is noted that by definition,
the transformation rule for the components of elements of $J^{k}\vb$,
may be written in the form
\begin{equation}
\vf_{,\mii'}^{\ga'}=G_{\mii'\ga}^{\ga'\mii}\vf_{,\mii}^{\ga},\qquad\abs{\mii}\le\abs{\mii'},\label{eq:Transf_Jet_Comp-4}
\end{equation}
where $G_{\mii'\ga}^{\ga'\mii}$ contains derivatives of the aforementioned
transformations. 

Using the transformation rules (\ref{eq:Transf_Jet_Comp-4}) and denoting
the Jacobian determinant of the transformation by $\Jac$, we may
write

\begin{equation}
\begin{split}\std_{\ga}^{\mii}\vf_{,\mii}^{\ga}\dee x & =\std_{\ga'}^{\mii'}\vf_{,\mii'}^{\ga'}\dee x',\qquad\abs{\mii}\le\abs{\mii'},\\
 & =\std_{\ga'}^{\mii'}G_{\mii'\ga}^{\ga'\mii}\vf_{,\mii}^{\ga}\dee x\Jac.
\end{split}
\end{equation}
The independence of the point values of the derivatives $\vf_{,\mii}^{\ga}$
imply that the transformation rule for the components $\std_{\ga}^{\mii}$
is
\begin{equation}
\std_{\ga}^{\mii}=\Jac\std_{\ga'}^{\mii'}G_{\mii'\ga}^{\ga'\mii},\qquad\abs{\mii}\le\abs{\mii'}\le k.
\end{equation}
We conclude that the transformation of the components of the variational
stress densities of order $l$ involves all components of equal or
higher order $l\le r\le k$. In particular, the statement that all
the components of order $r>l$ vanish, is invariant. This, of course,
gives meaning to a statement that a material is of order $k$ and
not any higher order.

\section{\label{sec:Hyper-Stresses,-the-Cauchy}Hyper-Stresses, the Cauchy
Approach}

\subsection{Body hyper-forces\label{subsec:Body-hyper-forces}}

A \emph{body hyper-force density} is an element of $\L{J^{k-1}\vb,\ext^{n}T^{*}\base}$.
A \emph{body force field} is a section of $\L{J^{k-1}\vb,\ext^{n}T^{*}\base}$.
Given a body force field $\bfc$ and a subbody $\reg\subset\base$,
the total power of the body hyper-force is
\begin{equation}
\fc_{\bfc}(\lv)=\int_{\reg}\bfc\cdot\lv,
\end{equation}
for any section $\lv$ of $J^{k-1}\vb$. In particular, for a section
$\vf$ of $\vb$, 
\begin{equation}
\fc_{\bfc}(j^{k-1}\vf)=\int_{\reg}\bfc\cdot j^{k-1}\vf.
\end{equation}

Using the notation introduced above, a body hyper-force field is represented
locally in the form
\begin{equation}
\bfc=\bfc_{\ga}^{\mi J}\sbase_{\mi J}^{\ga}\otimes\dee x,\qquad0\le\abs{\mi J}\le k-1.
\end{equation}
The action of a body force is 
\begin{equation}
\bfc\cdot\lv=\bfc_{\ga}^{\mi J}\lv_{\mi J}^{\ga}\dee x,\quad\text{and in particular,}\quad\bfc\cdot j^{k}\vf=\bfc_{\ga}^{\mi J}\vf_{,\mi J}^{\ga}\dee x,\quad0\le\abs{\mi J}\le k-1.
\end{equation}

Consider next the transformation rule for the components $\bfc_{\ga}^{\mi J}$.
Using (\ref{eq:Transf_Jet_Comp-4}), one has
\begin{equation}
\begin{split}\bfc_{\ga}^{\mi J}\vf_{,\mi J}^{\ga}\dee x & =\bfc_{\ga'}^{\mi J'}\vf_{,\mi J'}^{\ga'}\dee x',\qquad\abs{\mi J}\le\abs{\mi J'}\le k-1,\\
 & =\bfc_{\ga'}^{\mi J'}G_{\mi J'\ga}^{\ga'\mi J}\vf_{,\mi J}^{\ga}\dee x\Jac.
\end{split}
\end{equation}
The independence of the point values of the derivatives $\vf_{,\mii}^{\ga}$
imply that the transformation rule for the components $\std_{\ga}^{\mii}$
is
\begin{equation}
\bfc_{\ga}^{\mi J}=\Jac\bfc_{\ga'}^{\mi J'}G_{\mi J'\ga}^{\ga'\mi J},\qquad\abs{\mi J}\le\abs{\mi J'}\le k-1.
\end{equation}
Similarly to the variational stress, we conclude that the transformation
of the components of the body hyper-force densities of order $l$
involves all components of equal or higher order $l\le r\le k$. In
particular, the statement that all the components of order $r>l$
vanish, is invariant. .

\subsection{Hyper-traction fields\label{subsec:Surface-hyper-forces}}

Let $\reg\subset\base$ be a subbody. A \emph{surface hyper-force
density} or a \emph{hyper-traction }on $\bdry\reg$ is an element
of $\L{J^{k-1}\vb,\ext^{n-1}T^{*}\bdry\reg}$. Here, and in what follows,
we view $\L{J^{k-1}\vb,\ext^{n-1}T^{*}\bdry\reg}$ as a vector bundle
over $\bdry\reg$ and in our notation we omit the indication that
$J^{k-1}\vb$ is restricted to $\bdry\reg$. A \emph{hyper-traction
field} on $\bdry\reg$ is a section of $\L{J^{k-1}\vb,\ext^{n-1}T^{*}\bdry\reg}$.
Given a subbody $\reg\subset\base$ and a hyper-traction field $\sfc_{\reg}$
and, the total power of the hyper-traction is the action of the functional
\begin{equation}
\fc_{\sfc}(\lv)=\int_{\bdry\reg}\sfc_{\reg}\cdot\lv,
\end{equation}
for any section $\lv$ of $(J^{k-1}\vb)\resto{\bdry\reg}$. In particular,
for a section $\vf$ of $\vb$, 
\begin{equation}
\fc_{\sfc}(j^{k-1}\vf)=\int_{\bdry\reg}\sfc_{\reg}\cdot j^{k-1}\vf.
\end{equation}
It is noted that the jet of $\vf$ is taken relative to $\base$ and
not $\bdry\reg$.

Let $(z^{p})$, $p=1,\dots,n-1$ be local coordinates in $\bdry\reg$,
so that 
\begin{equation}
\dee z=\lisuppc{\dee z}1{\wedge}{n-1}
\end{equation}
is a local volume form. Then, in an adapted coordinate system, a hyper-traction
field is represented locally in the form
\begin{equation}
\sfc_{\reg}=\sfc_{\reg\ga}^{\mi J}\sbase_{\mi J}^{\ga}\otimes\dee z,\qquad0\le\abs{\mi J}\le k-1.
\end{equation}
The action of a hyper-traction is 
\begin{equation}
\sfc_{\reg}\cdot\lv=\sfc_{\reg\ga}^{\mi J}\lv_{\mi J}^{\ga}\dee z,\quad\text{and in particular,}\quad\sfc_{\reg}\cdot j^{k}\vf=\sfc_{\reg\ga}^{\mi J}\vf_{,\mi J}^{\ga}\dee z,\quad0\le\abs{\mi J}\le k-1.
\end{equation}

\subsection{Smooth force functionals and hyper-force systems}

We will say that a force functional on a subbody $\reg$ is smooth
if it is induced by a body hyper-force field $\bfc$ and a hyper-traction
field $\sfc_{\reg}$ in the form
\begin{equation}
\fc_{\reg}(\vf)=\int_{\reg}\bfc\cdot j^{k-1}\vf+\int_{\bdry\reg}\sfc_{\reg}\cdot j^{k-1}\vf\label{eq:Hyper-force_systems}
\end{equation}
for every virtual displacement $\vf:\base\to\vb$. Henceforth, we
will be interested only in smooth force functionals.

Let $\{\sfc_{\reg}\}$ be a collection of sections $\sfc_{\reg}$
of $L\bigp{\avb\resto{\bdry\reg},\ext^{n-1}T^{*}\bdry\reg}$ for all
subbodies $\reg\subset\base.$ We will refer to $\{\sfc_{\reg}\}$
as a \emph{system of hyper-tractions}. The collection $(\bfc,\{\sfc_{R}\})$
will be referred to as a \emph{smooth system of hyper-forces}, where
each force functional in $\{\fc_{\reg}$\} is represented as in the
equation above.

\subsection{Traction hyper-stresses\label{subsec:Traction-hyper-stresses-1}}

A \emph{traction hyper-stress} is an element 
\begin{equation}
\st_{0}\in\L{J^{k-1}\vb,\ext^{n-1}T^{*}\base}
\end{equation}
 and a \emph{traction hyper-stress field} is a smooth section $\tst$
of $\L{J^{k-1}\vb,\ext^{n-1}T^{*}\base}$. 

Let $\reg\subset\base$ be an $n$-dimensional submanifold with boundary
$\bdry\reg$. The inclusion
\begin{equation}
\incl_{\bdry\reg}:\bdry\reg\tto\base,
\end{equation}
induces the following diagram where $(\incl_{\bdry\reg})^{*}T\base\simeq T\base\resto{\bdry\reg}$
is the pullback of the tangent bundle and $\delta\incl_{\bdry\reg}$
is the induced vector bundle morphism over the boundary that makes
the diagram commutative\textemdash an inclusion of tangent vectors.

\[
\begin{xy}
(0,0)*+{\partial\reg} = "Q";
(40,0)*+{\base} = "S";
(40,25)*+{T\base} ="TS";
(-10,25)*+{T\partial\reg} = "TQ";
(15,25)*+{(\incl_{\bdry\reg})^*T\base} ="jTS";
{\ar@{->}_{\incl_{\bdry\reg}} "Q"; "S"};
{\ar@{->}_{\tau_{\bdry\reg}} "TQ"; "Q"};
{\ar@{->}^{\incl_{\bdry\reg}^*\tau_\base} "jTS"; "Q"};
{\ar@{->}^{\tau_\base} "TS"; "S"};
{\ar@{->}@/^{2pc}/^{T\incl_{\bdry\reg}} "TQ"; "TS"};
{\ar@{->}^{\delta\incl_{\bdry\reg}} "TQ"; "jTS"};
{\ar@{->}^{{\tau_\base^*}\incl_{\bdry\reg}} "jTS"; "TS"};
\end{xy}
\]

The dual vector bundle morphism, 
\begin{equation}
\rho_{\bdry\reg}=(\delta\incl_{\bdry\reg})^{*}:\left(\ext^{n-1}T^{*}\base\right)\resto{\bdry\reg}\simeq\ext^{n-1}\bigp{\incl_{\bdry\reg}^{*}T\base}^{*}\tto\ext^{n-1}T^{*}\bdry\reg,
\end{equation}
restricts alternating tensors to vectors tangent to the boundary and
it induces the pullback of forms.

Let $z\in\bdry\reg$, and $\st_{0}\in\L{J^{k-1}\vb,\ext^{n-1}T^{*}\base}_{z}$.
Then, 
\begin{equation}
\sfc_{0}=\rho_{\bdry\reg}\comp\tst_{0}\in\L{J^{k-1}\vb,\ext^{n-1}T^{*}\bdry\reg}_{z}\label{eq:Gen_Cauch}
\end{equation}
is a hyper-traction induced at $z$ by $\tst_{0}$. Similarly, a traction
hyper-stress field $\tst$ induces a hyper-traction field
\begin{equation}
\sfc=\rho_{\bdry\reg}\comp\tst=\incl_{\bdry\reg}^{\sharp}\tst.\label{eq:Gen_Cauchy-1}
\end{equation}
Thus, Equation (\ref{eq:Gen_Cauchy-1}) is a generalization of the
traditional Cauchy formula.

Let $v\contr\go$ denote the contraction (inner product) of the form
$\go$ with the tangent vector $v$. As $\{\bdry_{i}\contr\dee x\}$
may serve as a basis for $\ext^{n-1}T^{*}\base$, in analogy with
(\ref{eq:Rep_Var_Str}), a traction hyper-stress field may be represented
in the form
\begin{equation}
\st=\st_{\ga}^{\mi Ji}\dbase_{\mi J}^{\ga}\otimes(\partial_{i}\contr\dee x),\qquad\abs{\mi J}\le k-1.\label{eq:Rep_Tract_Str-1}
\end{equation}
For a section $\alv$ of $J^{k-1}\vb$, one has therefore,
\begin{equation}
\st\cdot\alv=\st_{\ga}^{\mi Ji}\alv_{\mi J}^{\ga}(\partial_{i}\contr\dee x),
\end{equation}
and in particular,
\begin{equation}
\st\cdot j^{k-1}\vf=\st_{\ga}^{\mi Ji}\vf_{,\mi J}^{\ga}(\partial_{i}\contr\dee x),\qquad\abs{\mi J}\le k-1.
\end{equation}

Consider next the transformation rule for the components $\st_{\ga}^{\mi Ji}$.
Using the transformations (\ref{eq:Transf_Jet_Comp-4}) and the invariance
of the action, we have
\begin{equation}
\begin{split}\tst_{\ga}^{\mii i}\vf_{,\mii}^{\ga}\bdry_{i}\contr\dee x & =\tst_{\ga'}^{\mii'i'}\vf_{,\mii'}^{\ga'}\bdry_{i'}\contr\dee x',\qquad\abs{\mii}\le\abs{\mii'},\\
 & =\tst_{\ga'}^{\mii'i'}G_{\mii'\ga}^{\ga'\mii}\vf_{,\mii}^{\ga}x_{,i'}^{i}\bdry_{i}\contr\dee x\Jac.
\end{split}
\end{equation}
Thus,
\begin{equation}
\tst_{\ga}^{\mii i}=\Jac\tst_{\ga'}^{\mii'i'}G_{\mii'\ga}^{\ga'\mii}x_{,i'}^{i},\qquad\abs{\mii}\le\abs{\mii'}.
\end{equation}
The last equation is a generalization of the classical relation between
the Cauchy stress and the first Piola-Kirchhoff stress, which as presented
here, is just a transformation of coordinates of the right geometric
object. Once again, the transformation of the components of the traction
hyper-stress densities of order $l$ involves all components of equal
or higher order $l\le r\le k-1$. 

Finally, given a traction hyper-stress, $\tst$, and a body hyper-force,
$\bfc$, one may write the force functional on each subbody as
\begin{equation}
\begin{split}\fc_{\reg}(\vf) & =\int_{\reg}\bfc\cdot j^{k-1}\vf+\int_{\bdry\reg}\tst\cdot j^{k-1}\vf,\\
 & =\int_{\reg}\bfc\cdot j^{k-1}\vf+\int_{\reg}\dee(\tst\cdot j^{k-1}\vf).
\end{split}
\label{eq:force_and_TracSt}
\end{equation}

\subsection{Cauchy's theorem and traction hyper-stresses}

We review here Cauchy's theorem on manifolds as in \cite{segev_cauchys_1999,segev_notes_2013}.
The Cauchy theorem asserts that under certain boundedness and locality
conditions, a system of forces on all subbodies induce a unique traction
stress field. The theorem is formulated for a general vector bundle
$\avb$ and the terminology used is for the case of standard continuum
mechanics, $k=1$. We will use it later for the case $\avb=J^{k-1}\vb$. 

It is assumed in the sequel that the collection of subbodies of $\base$
includes $n$-dimensional chains and in particular simplices. Furthermore,
it is assumed that a particular orientation has been chosen for the
manifold $\base$.
\begin{defn}
We will say that a system $\{\sfc_{\reg}\}$ of tractions is \emph{consistent}
if the following conditions are satisfied.
\begin{enumerate}
\item \textbf{\textit{\emph{Boundedness.}}}\textit{\emph{ There is a section
$\xi$ of $\L{J^{1}U,\ext^{n}T^{*}\base}$ such that for each $\reg$,}}\emph{
\begin{equation}
\abs{\int_{\bdry\reg}\sfc_{\reg}(\avf\resto{\bdry\reg})}\les\int_{\reg}\abs{\xi(j^{1}\avf)},\label{eq:BoundednessStresses}
\end{equation}
for every smooth section $\avf$ of $\avb$. }
\item \textbf{\textit{\emph{Cauchy's postulate of locality.}}}\textit{\emph{
Let $v_{1},\dots,v_{n-1}\in T_{x}\base$ be a collection of tangent
vector at $x\in\base$. Let $\reg$ be any subbody such that $v_{1},\dots,v_{n-1}\in T_{x}\bdry\reg$,
and the collection of vectors is positively oriented relative to the
orientation of $\bdry\reg$. Then, $\sfc_{\reg}(\avf)(v_{1},\dots,v_{n-1})$
is independent of $\reg$. Thus, there is an alternating multilinear
mapping 
\begin{equation}
\vcmap_{x}:(T_{x}\base)^{n-1}\tto U_{x}^{*}
\end{equation}
such that
\begin{equation}
\sfc_{\reg}(\avf)(v_{1},\dots,v_{n-1})=\vcmap_{x}(v_{1},\dots,v_{n-1})(\avf)\fall u\in\avb_{x}.
\end{equation}
}}
\item \textbf{\textit{\emph{Regularity.}}}\textit{\emph{ The mapping $x\mapsto\vcmap_{x}$
is smooth.}}
\end{enumerate}
A smooth system of forces $(\bfc,\{\sfc_{\reg}\})$ is said to be
consistent if $\{\sfc_{\reg}\}$ is a consistent system of tractions.
\end{defn}

\begin{thm}
Let $\{\sfc_{\reg}\}$ be a consistent force system, then, there is
a unique traction stress field $\tst$ such that for each subbody
$\reg$, the surface force is given by the Cauchy formula
\begin{equation}
\sfc_{\reg}=\rho_{\bdry\reg}\comp\tst.\label{eq:Cauchy_Thm}
\end{equation}
\end{thm}

\selectlanguage{english}%
\foreignlanguage{american}{}

\selectlanguage{american}%
Cauchy's theorem, in the generalized setting outlined above, applies
readily to traction hyper-stresses by setting $\avb=J^{k-1}\vb$.
Assuming that the system of hyper-tractions satisfies the conditions
above, it follows from Cauchy's theorem that there is a unique traction
hyper-stress that induces it using (\ref{eq:Cauchy_Thm}). 

\section{Non-Holonomic Variational Stresses\label{sec:Non-Holonomic-Variational-Stress}}

Unlike the simple case, $k=1$, variational hyper-stresses do not
determine uniquely traction hyper-stresses for higher order theories.
Another mathematical object is needed\textemdash the non-holonomic
stress.

\subsection{The exterior jet and non-holonomic variational stresses}

\begin{proposition}

There is a natural linear, first order-differential operator, the
exterior jet $\coj$, taking sections of $\L{J^{k-1}\vb,\ext^{n-1}T^{*}\base}$
into sections of $\L{J^{1}(J^{k-1}\vb),\ext^{n}T^{*}\base}$, defined
by the condition
\begin{equation}
\coj\tst\cdot j^{1}\lv=\dee(\tst\cdot\lv)\label{eq:Def-Ext_Jet}
\end{equation}
for every section $\lv$ of $J^{k-1}\vb$.

\end{proposition}
\begin{proof}
Let $\tst$ be a traction hyper-stress field represented as in Equation
(\ref{eq:Rep_Tract_Str-1}) and let a generic section $\lv$ of $J^{k-1}\vb$
be given locally by $\lv=\lv_{\mi J}^{\ga}\sbase_{\ga}^{\mi J}$.
One has
\begin{equation}
\dee(\st\cdot\lv)=(\st_{\ga,i}^{\mi Ji}\lv_{\mi J}^{\ga}+\st_{\ga}^{\mi Ji}\lv_{\mi J,i}^{\ga})\dee x,\qquad\abs{\mi J}\le k-1.
\end{equation}
It follows that $\dee(\tst\cdot\lv)$ is an $n$-form that depends
pointwise on the values of the representatives of $\lv$ and their
first derivatives, \ie,  it depends on $j^{1}\lv$. Hence, there
is a unique section $\coj\tst$ of $\L{J^{1}(J^{k-1}\vb),\ext^{n}T^{*}\base}$,
for which condition (\ref{eq:Def-Ext_Jet}) holds. The value of the
section $\coj\tst$ at a point $x\in\base$ depends linearly on the
first jet at $x$ of the section $\tst:\base\to\L{J^{k-1}\vb,\ext^{n-1}T^{*}\base}$.
Hence, $\coj$ is a first order linear differential operator.
\end{proof}
\begin{defn}
An element of $\L{J^{1}(J^{k-1}\vb),\ext^{n}T^{*}\base}$ will be
referred to as a \emph{non-holonomic (variational hyper-) stress.}
A section of $\L{J^{1}(J^{k-1}\vb),\ext^{n}T^{*}\base}$ is a \emph{non-holonomic
stress field.}
\end{defn}

Thus, $\coj\tst$ is a non-holonomic stress field. As noted, the terminology
follows from the fact that a section of $J^{k-1}\vb$ need not be
the $k$-th jet of a section of $\vb$.

\subsection{Local representation of non-holonomic stresses and exterior jets}

An element $\alv\in J^{1}(J^{k-1}\vb)$ is represented locally in
the form
\begin{equation}
\alv=\alv_{\mi J}^{\ga}\sbase_{\ga}^{\mi J}+\alv_{\mi Jj}^{\ga}\sbase_{\ga}^{\mi J}\otimes\dee x^{j},
\end{equation}
where it is noted that while $\alv_{\mi J}^{\ga}$ and $\alv_{\mi Jj}^{\ga}$
are symmetric under permutations of $\mi J$, $\alv_{\mi Jj}^{\ga}$
need not be symmetric under permutations of $\mi Jj$. For a section
$\lv$ of $J^{k-1}\vb$, represented locally in the form $\lv=\lv_{\mi J}^{\ga}\sbase_{\ga}^{\mi J}$,
one has locally
\begin{equation}
j^{1}\lv=\lv_{\mi J}^{\ga}\sbase_{\ga}^{\mi J}+\lv_{\mi J,j}^{\ga}\sbase_{\ga}^{\mi J}\otimes\dee x^{j},
\end{equation}
and for a section $\vf$ of $\vb$,
\begin{equation}
j^{1}(j^{k-1}\vf)=\vf_{,\mi J}^{\ga}\sbase_{\ga}^{\mi J}+w_{,\mi Jj}^{\ga}\sbase_{\ga}^{\mi J}\otimes\dee x^{j}.
\end{equation}
Evidently, $j^{1}(j^{k-1}\vf)$ represents a holonomic element of
$J^{1}(J^{k-1}\vb)$ for which complete symmetry holds.

It follows that a non-holonomic stress should be locally of the form
\begin{equation}
\nhs=\bigp{\nhsa_{\ga}^{\mi J}\sbase_{\mi J}^{\ga}+\nhsb_{\ga}^{\mi Jj}\sbase_{\mi J}^{\ga}\otimes\bdry_{j}}\otimes\dee x.
\end{equation}
Its action is given by
\begin{equation}
\nhs(\alv)=\bigp{\nhsa_{\ga}^{\mi J}\alv_{\mi J}^{\ga}+\nhsb_{\ga}^{\mi Jj}\alv_{\mi Jj}^{\ga}}\dee x,\qquad\nhs(j^{1}\lv)=\bigp{\nhsa_{\ga}^{\mi J}\lv_{\mi J}^{\ga}+\nhsb_{\ga}^{\mi Jj}\lv_{\mi J,j}^{\ga}}\dee x,
\end{equation}
and
\begin{equation}
\nhs(j^{1}(j^{k-1}\vf))=\bigp{\nhsa_{\ga}^{\mi J}\vf_{,\mi J}^{\ga}+\nhsb_{\ga}^{\mi Jj}\vf_{,\mi Jj}^{\ga}}\dee x.
\end{equation}

Given a traction hyper-stress field $\tst$, the exterior jet field
should be of the form
\begin{equation}
\coj\tst=\bigp{(\coj\tst)_{\ga}^{\mi J}\sbase_{\mi J}^{\ga}+(\underline{\coj\tst})_{\ga}^{\mi Jj}\sbase_{\mi J}^{\ga}\otimes\bdry_{j}}\otimes\dee x.
\end{equation}
The definition of the exterior jet implies that for any section $\lv$
of $J^{k-1}\vb$, $\coj\tst\cdot j^{1}\lv=\dee(\tst\cdot\lv)$, and
so
\begin{equation}
(\coj\tst)_{\ga}^{\mi J}\lv_{\mi J}^{\ga}+(\underline{\coj\tst})_{\ga}^{\mi Jj}\lv_{\mi J,j}^{\ga}=\tst_{\ga,j}^{\mi Jj}\lv_{,\mi J}^{\ga}+\tst_{\ga}^{\mi Jj}\lv_{\mi J,j}^{\ga}.
\end{equation}
At any point $x$, the values of the components $\lv_{\mi J}^{\ga}$
and $\lv_{\mi J,j}^{\ga}$ are independent. Thus,
\begin{equation}
(\coj\tst)_{\ga}^{\mi J}=\tst_{\ga,j}^{\mi Jj},\qquad(\underline{\coj\tst})_{\ga}^{\mi Jj}=\tst_{\ga}^{\mi Jj},
\end{equation}
\begin{equation}
\coj\tst=\bigp{\tst_{\ga,j}^{\mi Jj}\sbase_{\mi J}^{\ga}+\tst_{\ga}^{\mi Jj}\sbase_{\mi J}^{\ga}\otimes\bdry_{j}}\otimes\dee x.\label{eq:Exterior_Jet_J11}
\end{equation}

\subsection{The non-holonomic stress induced by a consistent hyper-force system}

We now consider systems consisting of body hyper-forces and hyper-tractions
as in (\ref{eq:Hyper-force_systems}). We will say that the hyper-force
system is consistent if it satisfies the conditions of consistency
of the generalized Cauchy theorem for the case $\avb=J^{k-1}\vb$. 

\begin{proposition}\label{prop:Induced_Non_Holonomic}

A smooth consistent hyper-force system\foreignlanguage{english}{ $(\bfc,\{\sfc_{\reg}\})$}
induces a unique non-holonomic stress field $\nhs$ which represents
it by
\begin{equation}
\fc_{\reg}(\vf)=\int_{\reg}\nhs\cdot j^{1}(j^{k-1}\vf).\label{eq:Rep_Syst_NonHolSt}
\end{equation}

\end{proposition}
\begin{proof}
A consistent hyper-force system induces, by the generalized Cauchy
theorem for the case $\avb=J^{k-1}\vb$, a unique traction hyper-stress
$\tst$. Thus, Equation (\ref{eq:force_and_TracSt}) applies, and
using the definition of the exterior jet, we may rewrite it as
\begin{equation}
\fc_{\reg}(\vf)=\int_{\reg}\bfc\cdot j^{k-1}\vf+\coj\tst\cdot j^{1}(j^{k-1}\vf).\label{eq:Rep_Force_Ext_Jet_Sigma}
\end{equation}

Observing Equation (\ref{eq:Rep_Force_Ext_Jet_Sigma}), the value
of the integrand at any $x\in\base$ is an $n$-alternating tensor
that is evidently a linear function of $j^{1}(j^{k-1}\vf)(x)$. For
a section $\lv$ of $J^{k-1}\vb$, set
\begin{equation}
\nhs\cdot j^{1}\lv=\bfc\cdot\lv+\coj\tst\cdot j^{1}\lv.\label{eq:Define_Non_Hol_St}
\end{equation}
Thus, $\nhs$ is a section of $\Ljj$ such that 
\begin{equation}
\fc_{\reg}(\vf)=\int_{\reg}\bigp{\bfc\cdot j^{k-1}\vf+\coj\tst\cdot j^{1}(j^{k-1}\vf)}=\int_{\reg}\nhs\cdot j^{1}(j^{k-1}\vf).
\end{equation}
In fact, $\nhs$ induces a linear functional $\wh{\nhs}_{\reg}$,
which extends $\fc_{\reg}$ to act on all sections of $J^{k-1}\vb$
by
\begin{equation}
\wh{\nhs}_{\reg}(\lv)=\int_{\reg}\bigp{\bfc\cdot\lv+\coj\tst\cdot j^{1}\lv}=\int_{\reg}\nhs\cdot j^{1}\lv.
\end{equation}
The local representation of $\nhs$ may be readily obtained from (\ref{eq:Exterior_Jet_J11})
as 
\begin{gather}
\nhs=\bigp{(\tst_{\ga,j}^{\mi Jj}+\bfc_{\ga}^{\mi J})\sbase_{\mi J}^{\ga}+\tst_{\ga}^{\mi Jj}\sbase_{\mi J}^{\ga}\otimes\bdry_{j}}\otimes\dee x\qquad0\le\abs{\mi J}\le k-1,\label{eq:Rep_Induced_NHS}\\
\nhsa_{\ga}^{\mi J}=\tst_{\ga,j}^{\mi Jj}+\bfc_{\ga}^{\mi J},\qquad\nhsb_{\ga}^{\mi Jj}=\tst_{\ga}^{\mi Jj}.
\end{gather}
\end{proof}
\begin{proposition}\label{prop:NHS_deter_traction}

There is a natural mapping 
\begin{equation}
p_{\tst}:\Ljj\tto\L{J^{k-1}\vb,\ext^{n-1}\base},
\end{equation}
whereby a non-holonomic stress field $\nhs$ induces a unique traction
hyper-stress $\tst$.

\end{proposition}
\begin{proof}
The local expression for the action of $\nhs$ is
\begin{equation}
\nhs\cdot j^{1}\lv=\bigp{\nhs_{\ga}^{\mi J}\lv_{\mi J}^{\ga}+\nhsb_{\ga}^{\mi Jj}\lv_{\mi J,j}^{\ga}}\otimes\dee x,\qquad0\le\abs{\mi J}\le k-1.\label{eq:Action_NHS_on_Jk-1}
\end{equation}
Restricting $\nhs(x)$ to sections for which $\lv_{\mi J}^{\ga}(x)=0$
for all $\mi J$, $\abs{\mi J}\le k-1$, \ie, to elements of 
\begin{equation}
V^{0}J^{1}(J^{k-1}\vb)=\kernel\pi_{0}^{1},
\end{equation}
the action of $\nhs$ determines $\tst$ uniquely. In other words,
the inclusion $\incl_{1}^{0}$ of the vertical bundle induces 
\begin{equation}
\begin{split}\incl_{1}^{0*}:\Ljj\tto & \L{VJ^{1}(J^{k-1}\vb),\ext^{n}T^{*}\base},\\
 & \isom\L{L(T\base,J^{k-1}\vb),\ext^{n}T^{*}\base},\\
 & \isom T\base\otimes(J^{k-1}\vb)^{*}\otimes\ext^{n}T^{*}\base,\\
 & \isom(J^{k-1}\vb)^{*}\otimes\ext^{n-1}T^{*}\base,\\
 & \isom\L{J^{k-1}\vb,\ext^{n-1}\base},
\end{split}
\end{equation}
where we used the natural isomorphism $VJ^{1}\avb\isom L(T\base,\avb)$
for any vector bundle $\avb$ over $\base$, and the isomorphism 
\begin{equation}
T\base\otimes\ext^{n}T^{*}\base\tto\ext^{n-1}T^{*}\base,\qquad(v,\theta)\lmt v\contr\theta.
\end{equation}
Thus, locally, 
\begin{equation}
\tst_{\ga}^{\mi Jj}=\nhsb_{\ga}^{\mi Jj}.\label{eq:Components_of_psigma(P)}
\end{equation}
\end{proof}
\begin{proposition}

There is a natural linear differential operator $\diver$ mapping
jets of non-holonomic stress fields into body hyper-force fields.
The differential operator $\diver$ is defined by 
\begin{equation}
\diver\nhs=\coj\comp p_{\st}(\nhs)-\nhs.\label{eq:Divergence_definition}
\end{equation}
The non-holonomic stress $\nhs$ represents the force system\foreignlanguage{english}{
$(\bfc,\{\sfc_{\reg}\})$} if an only if
\begin{equation}
p_{\st}\comp\nhs=\tst,\qquad\text{and}\qquad\diver\nhs+\bfc=0.\label{eq:Stress_and_Body_Force}
\end{equation}

\end{proposition}
\begin{proof}
To prove that $\diver\nhs$ is a section of $\L{J^{k-1}\vb,\ext^{n}T^{*}\base}$,
we use (\ref{eq:Components_of_psigma(P)}) and (\ref{eq:Exterior_Jet_J11}).
Thus, locally, for $0\le\abs{\mi J}\le k-1$,
\begin{equation}
\begin{split}\diver\nhs & =\bigp{\nhsb_{\ga,j}^{\mi Jj}\sbase_{\mi J}^{\ga}+\nhsb_{\ga}^{\mi Jj}\sbase_{\mi J}^{\ga}\otimes\bdry_{j}-\nhsa_{\ga}^{\mi J}\sbase_{\mi J}^{\ga}-\nhsb_{\ga}^{\mi Jj}\sbase_{\mi J}^{\ga}\otimes\bdry_{j}}\otimes\dee x,\\
 & =(\nhsb_{\ga,j}^{\mi Jj}-\nhsa_{\ga}^{\mi J})\sbase_{\mi J}^{\ga}\otimes\dee x.
\end{split}
\label{eq:Divergence_Local_Rep}
\end{equation}
In view of Propositions \ref{prop:Induced_Non_Holonomic} and \ref{prop:NHS_deter_traction},
it remains to show that $\bfc=-\diver\nhs$. This follows, however,
from (\ref{eq:Define_Non_Hol_St}) and (\ref{eq:Divergence_definition}).
\end{proof}

\section{Induced Variational Stresses\label{sec:Induced-Variational-Stresses}}

One could naively assume that just as in the case $k=1$, a variational
hyper-stress, a section of $\L{J^{k}\vb,\ext^{n}T^{*}\base}$ will
induce a unique traction hyper-stress. This, however, is not the case.
It is sufficient to observe that the in the local expression for the
variational hyper-stress $\std_{\ga}^{\mi I}\sbase_{\mii}^{\ga}\otimes\dee x$,
$\abs{\mii}\le k$, the components $\std_{\ga}^{\mii}$ are symmetric
under permutations of $\mii$. On the other hand, in the local expression
for the traction hyper-stress $\st_{\ga}^{\mi Ji}\dbase_{\mi J}^{\ga}\otimes(\partial_{i}\contr\dee x)$,
$\abs{\mi J}\le k-1$, the arrays $\tst^{\mi Ji}$, which are of the
same order as $\std_{\ga}^{\mi I}$, need not be symmetric under all
permutations of $\mi Jj$. In other words, unlike the non-holonomic
stress that is in one-to-one correspondence with a force system, the
variational hyper-stress encodes less information and cannot be used
to determine the hyper-traction on the boundaries of subbodies. The
present section studies these issues with some detail.

\subsection{The variational hyper-stress induced by a non-holonomic stress\label{subsec:VarST_IndBy_NHD}}

Let
\begin{equation}
\incl_{1,k-1}^{k}:J^{k}\vb\tto J^{1}(J^{k-1}\vb),\qquad j^{k}\vf(x)\lmt j^{1}(j^{k-1}\vf)(x),
\end{equation}
be the natural inclusion of holonomic elements. Then, the dual mapping
\begin{equation}
\incl_{1,k-1}^{k*}:\Ljj\tto\L{J^{k}\vb,\ext^{n}T^{*}\base},\qquad\nhs\lmt\nhs\comp\incl_{1,k-1}^{k},
\end{equation}
assigns variational hyper-stresses to non-holonomic stresses. Let
$\std=\incl_{1,k-1}^{k*}(\nhs)$. The local representations of the
actions of $\std$ and $\nhs$ imply that for any section $\vf$ of
$\vb$,
\begin{alignat*}{2}
\std_{\ga}^{\mii}\vf_{,\mii}^{\ga} & =\nhsa_{\ga}^{\mi J}\vf_{,\mi J}^{\ga}+\nhsb_{\ga}^{\mi Jj}\vf_{,\mi Jj}, & 0\le\abs{\mii}\le k,\quad0\le\abs{\mi J}\le k-1,\\
 & =\nhsa_{\ga}\vf^{\ga}+\tm{\nhs}_{\ga}^{\mi Lj}\vf_{,\mi Lj}^{\ga}+\nhsb_{\ga}^{\mi Lj}\vf_{,\mi Lj}^{\ga}+\itm{\nhsb}_{\ga}^{\mi K}\vf_{,\mi K}^{\ga}, & 0\le\abs{\mi L}\le k-2,\quad\abs{\mi K}=k,\\
 & =\nhsa_{\ga}\vf^{\ga}+(\tm{\nhs}_{\ga}^{\mi Lj}+\nhsb_{\ga}^{\mi Lj})\vf_{,\mi Lj}^{\ga}+\itm{\nhsb}_{\ga}^{\mi K}\vf_{,\mi K}^{\ga},\\
 & =\nhsa_{\ga}\vf^{\ga}+(\nhs_{\ga}^{\mi J}+\itm{\nhsb}_{\ga}^{\mi J})\vf_{,\mi J}^{\ga}+\itm{\nhsb}_{\ga}^{\mi K}\vf_{,\mi K}^{\ga}, & 1\le\abs{\mi J}\le k-1,\quad\abs{\mi K}=k.
\end{alignat*}
As the relations above should hold for any point $x\in\base$ and
since compatibility imposes no restrictions on the point values of
the various derivatives, we conclude that $\std=\incl_{1,k-1}^{k*}(\nhs)$
is represented by
\begin{equation}
\begin{alignedat}{2}\std_{\ga}^{\mi K} & =\itm{\nhsb}_{\ga}^{\symm{\mi K}}, & \qquad\abs{\mi K}=k,\\
\std_{\ga}^{\mi J} & =\nhsa_{\ga}^{\mi J}+\itm{\nhsb}_{\ga}^{\symm{\mi J}}, & \qquad1\le\abs{\mi J}<k,\\
\std_{\ga} & =\nhsa_{\ga}.
\end{alignedat}
\label{eq:Rep_dual_ijj}
\end{equation}

\subsection{The reduced exterior jet}

The \emph{reduced exterior jet} of a traction hyper-stress $\tst$
is defined as
\begin{equation}
\coj_{k}\tst:=\incl_{1,k-1}^{k*}(\coj\tst),
\end{equation}
that is, 
\begin{equation}
\begin{split}\coj_{k}\tst\cdot j^{k}\vf & =\coj\tst\cdot j^{1}(j^{k-1}\vf)\end{split}
.\label{eq:Def_Reduced_Cojet}
\end{equation}

Using (\ref{eq:Rep_dual_ijj}) and (\ref{eq:Exterior_Jet_J11}) one
obtains immediately,
\begin{equation}
\begin{alignedat}{2}(\coj_{k}\tst)_{\ga}^{\mi K} & =\itm{\tst}_{\ga}^{\symm{\mi K}}, & \qquad\abs{\mi K}=k,\\
(\coj_{k}\tst)_{\ga}^{\mi J} & =\tst_{\ga,j}^{\mi Jj}+\itm{\tst}_{\ga}^{\symm{\mi J}}, & \qquad1\le\abs{\mi J}<k,\\
(\coj_{k}\tst)_{\ga} & =\tst_{\ga,j}^{j}.
\end{alignedat}
\end{equation}
Evidently, one can obtain these relations directly from (\ref{eq:Def_Reduced_Cojet}).

\subsection{The variational hyper-stress induced by a consistent hyper-force
system}

Let $(\bfc,\{\sfc_{\reg}\})$ be a consistent hyper-force system corresponding
to the traction hyper-stress $\tst$. Then, the corresponding unique
non-holonomic stress, as in (\ref{eq:Rep_Induced_NHS}), determines
through (\ref{eq:Rep_dual_ijj}) a unique variational hyper-stress
$\std$ given locally by
\begin{equation}
\begin{alignedat}{2}\std_{\ga}^{\mi K} & =\itm{\tst}_{\ga}^{\symm{\mi K}}, & \qquad\abs{\mi K}=k,\\
\std_{\ga}^{\mi J} & =\tst_{\ga,j}^{\mi Jj}+\bfc_{\ga}^{\mi J}+\itm{\tst}_{\ga}^{\symm{\mi J}}, & \qquad1\le\abs{\mi J}<k,\\
\std_{\ga} & =\tst_{\ga,j}^{j}+\bfc_{\ga}.
\end{alignedat}
\label{eq:VarSt_DeterBy_ForceSys}
\end{equation}
As mentioned above, only the symmetrized components of $\tst$ determine
$\std$. Conversely, specifying $\std$ and $\bfc$ is not sufficient
in order to determine a unique traction hyper-stress.

One may be tempted to postulate a-priori that the components of $\tst$
be symmetric. However, as shown below, symmetry of the components
of elements of $\L{J^{k-1}\vb,\ext^{n-1}T^{*}\base}$ is not an invariant
property.

\begin{example}
\label{exa:k=00003D2-1}Consider the case where $\vb=\base\times\reals$
and $k-1=3$. Let $x^{i'}=x^{i'}(x^{i})$ be a local coordinate transformation.
Then,
\begin{equation}
\begin{split}\vf_{,i'} & =\vf_{,i}x_{,i'}^{i},\\
\vf_{,i'j'} & =\vf_{,ij}x_{,i'}^{i}x_{,j'}^{j}+\vf_{,i}x_{,i'j'}^{i},\\
\vf_{,i'j'k'} & =\vf_{,ijk}x_{,i'}^{i}x_{,j'}^{j}x_{,k'}^{k}+\vf_{,ij}x_{,i'k'}^{i}x_{,j'}^{j}+\vf_{,ij}x_{,i'}^{i}x_{,j'k'}^{j}\\
 & \qquad+\vf_{,ik}x_{,i'j'}^{i}x_{,k'}^{k}+\vf_{,i}x_{,i'j'k'}^{i}.
\end{split}
\end{equation}
Invariance of the action imposes the condition that
\begin{equation}
\begin{split}\tst^{\mii l}\vf_{,\mii}\bdry_{l}\contr\dee x & =\tst^{\mii'l'}\vf_{,\mii'}\bdry_{l'}\contr\dee x',\qquad\abs{\mii}\le3,\\
 & =\tst^{\mii'l'}\vf_{,\mii'}x_{,l'}^{l}\bdry_{l}\contr\dee x\Jac,
\end{split}
\end{equation}
where $\Jac$ is the Jacobian determinant. After substitution of the
transformed derivatives, and in view of the arbitrariness of the values
of $\vf_{,\mii}$, we obtain
\begin{equation}
\begin{split}\tst^{ijkl}/\Jac & =\tst^{i'j'k'l'}x_{,i'}^{i}x_{,j'}^{j}x_{,k'}^{k}x_{,l'}^{l},\\
\tst^{ijl}/\Jac & =\tst^{i'j'l'}x_{,i'}^{i}x_{,j'}^{j}x_{,l'}^{l}+\tst^{i'j'k'l'}(x_{,i'k'}^{i}x_{,j'}^{j}x_{,l'}^{l}\\
 & \qquad\qquad\qquad+x_{,i'}^{i}x_{,j'k''}^{j}x_{,l'}^{l}+x_{,i'j'}^{i}x_{,k'}^{k}x_{,l'}^{l}),\\
\tst^{il}/\Jac & =\tst^{i'l'}x_{,i'}^{i}x_{,l'}^{l}+\tst^{i'j'l'}x_{,i'j'}^{i}x_{,l'}^{l}+\tst^{i'j'k'l'}x_{,i'j'k'}^{i}x_{,l'}^{l},\\
\tst^{l}/\Jac & =\tst^{l'}x_{,l'}^{l}.
\end{split}
\label{eq:-Example_k=00003D4-1}
\end{equation}
We note that in case $\tst^{i'j'k'l'}$ is symmetric, so is $\tst^{ijkl}$.
However, even if all the components $\tst^{\mii'}$ are symmetric,
this need not hold for $\tst^{ijl}$ and $\tst^{il}$. 

We conclude that symmetry of the components of a traction hyper-stress
is not an invariant property.
\end{example}

\begin{example}
It is observed that for the case $k=2$, so that $\tst^{ijl}=\tst^{ijkl}=0$
in the last equation, $\st^{il}$ is symmetric if and only if $\tst^{i'l'}$
is symmetric. Hence, for $k=2$, one may impose the condition that
$\tst^{ij}$ is symmetric. This, together with the condition $\bfc_{\ga}^{i}=0$,
imply that the variational stress determines a unique traction hyper-stress
and a unique body hyper-force.
\end{example}

In light of the foregoing discussion we make the following
\begin{defn}
A variational hyper-stress $\std$ is said to be \emph{consistent}
with a non-holonomic stress $\nhs$ if $\std=\incl_{1,k-1}^{k*}(\nhs)$.
\end{defn}

\begin{proposition}

If a variational hyper-stress $\std$ is consistent with a non-holonomic
stress $\nhs$ representing a hyper-force system\foreignlanguage{english}{
}$\{\fc_{\reg}\}=(\bfc,\{\sfc_{\reg}\})$, then, although $\bfc$
and $\{\sfc_{\reg}\}$ cannot be determined by $\std$, the force
functionals $\{\fc_{\reg}\}$ can be determined by $\std$ in the
form
\begin{equation}
\fc_{\reg}(\vf)=\int_{\reg}\std\cdot j^{k}\vf.
\end{equation}
In other words, while $\std$ represents the force system $\{\fc_{\reg}\}$
as in Sections \ref{subsec:Forces_and_Stress_Repr}, \ref{subsec:Smooth-variational-hyper-stresse},
it does not determine the corresponding body hyper-force and hyper-tractions.

\end{proposition}
\begin{proof}
The assertion follows immediately from the fact that consistency implies
that
\begin{equation}
\fc_{\reg}(\vf)=\int_{\reg}\nhs\cdot j^{1}(j^{k-1}\vf)=\int_{\reg}\std\cdot j^{k}\vf.
\end{equation}
\end{proof}

\section{Spaces of Hyper-Stresses and Constitutive Relations\label{sec:Spaces-of-Hyper-Stresses}}

Hyper-stresses have been defined above in the context of some particular
configuration $\conf$. To formulate constitutive relations, one has
to identify the appropriate spaces to which hyperstresses belong for
all possible configurations. The present section is concerned with
these issues.

\subsection{The space of variational hyper-stress density values\label{subsec:strDensVals}}

We first recall the following two natural isomorphisms. As fiber bundles
over $\base$,
\begin{equation}
VJ^{k}\fb\simeq J^{k}V\fb,
\end{equation}
and for any section $\conf:\base\to\fb$,
\begin{equation}
(j^{k}\conf)^{*}VJ^{k}\fb\simeq J^{k}(\conf^{*}V\fb)\label{eq:VJ=00003DJV-isom}
\end{equation}
(see \cite[Theorem 4.4.1]{saunders_geometry_1989} and \cite[Theorem. 17.1, p. 82]{palais_foundations_1968}). 

For a configuration $\conf$ and a point $x\in\base$, consider the
value $\std(x)$ of some variational hyper-stress field at the configuration
$\conf$. Using the isomorphism above and recalling (\ref{eq:VYisW}),
one has,
\begin{equation}
\begin{split}\std(x)\in & \;\L{(J^{k}\vb)_{x},\left(\ext^{n}T^{*}\base\right)_{x}},\\
 & =\L{(J^{k}(\conf^{*}V\fb)_{x},\left(\ext^{n}T^{*}\base\right)_{x}},\\
 & =\L{((j^{k}\conf)^{*}VJ^{k}\fb)_{x},\left(\ext^{n}T^{*}\base\right)_{x}},\\
 & =\L{(VJ^{k}\fb)_{j^{k}\conf(x)},\pi^{k*}\left(\ext^{n}T^{*}\base\right)_{j^{k}\conf(x)}},\\
 & =\L{VJ^{k}\fb,\pi^{k*}\ext^{n}T^{*}\base}_{j^{k}\conf(x)},\\
 & =\left[(j^{k}\conf)^{*}\L{VJ^{k}\fb,\pi^{k*}\ext^{n}T^{*}\base}\right]_{x}.
\end{split}
\label{eq:Space_of_VarStresses}
\end{equation}
One concludes that the values of variational hyper-stress fields may
always be identified with elements of the vector bundle 
\begin{equation}
\pi_{\spstd}:\spstd:=\L{VJ^{k}\fb,\pi^{k*}\ext^{n}T^{*}\base}\tto J^{k}\fb,
\end{equation}
independently of the configuration under consideration. For any configuration
$\conf$, variational hyper-stress densities are valued in 
\begin{equation}
(j^{k}\conf)^{*}\spstd\simeq\L{J^{k}(\conf^{*}V\fb),\ext^{n}T^{*}\base}.
\end{equation}

Intuitively, sections of $J^{k}\fb$ may be thought of as ``local
configurations'' and sections of $VJ^{k}\fb$ may be conceived as
variations thereof or as ``local velocity fields''. Thus, following
this line of thought, variational hyper-stresses may be thought of
as ``local forces''.

\subsection{The space of traction hyper-stress values}

In analogy with Section \ref{subsec:strDensVals}, we identify here
the space where traction hyper-stress fields assume their values,
independently of a particular configuration. Thus, for a configuration
$\conf$ and a point $x\in\base$, we consider the value $\tst(x)$
of some traction hyper-stress field at the configuration $\conf$.
One has,
\begin{equation}
\begin{split}\tst(x)\in & \;\L{(J^{k-1}\vb)_{x},\left(\ext^{n-1}T^{*}\base\right)_{x}},\\
 & =\L{(J^{k-1}(\conf^{*}V\fb)_{x},\left(\ext^{n-1}T^{*}\base\right)_{x}},\\
 & =\L{((j^{k-1}\conf)^{*}VJ^{k-1}\fb)_{x},\left(\ext^{n-1}T^{*}\base\right)_{x}},\\
 & =\L{(VJ^{k-1}\fb)_{j^{k-1}\conf(x)},\pi^{k-1*}\left(\ext^{n-1}T^{*}\base\right)_{j^{k-1}\conf(x)}},\\
 & =\L{VJ^{k-1}\fb,\pi^{k-1*}\ext^{n-1}T^{*}\base}_{j^{k-1}\conf(x)},\\
 & =\left[(j^{k-1}\conf)^{*}\L{VJ^{k-1}\fb,\pi^{k-1*}\ext^{n-1}T^{*}\base}\right]_{x}.
\end{split}
\end{equation}
It is concluded that the values of traction hyper-stress fields may
be identified with elements of the vector bundle
\begin{equation}
\pi_{\sptst}:\sptst:=\L{VJ^{k-1}\fb,\pi^{k-1*}\ext^{n-1}T^{*}\base}\tto J^{k-1}\fb.
\end{equation}

For a given configuration $\conf$, traction hyper-stresses are valued
in 
\begin{equation}
(j^{k-1}\conf)^{*}\sptst\simeq\L{(j^{k-1}\conf)^{*}VJ^{k-1}\fb,\ext^{n-1}T^{*}\base}.
\end{equation}

\subsection{The space of values of non-holonomic stresses}

Let 
\begin{equation}
\lc:\base\tto J^{k-1}\fb
\end{equation}
 be a smooth section and let $x\in\base$. Than, the value of a non-holonomic
stress field $\nhs$ at $x$ satisfies
\begin{equation}
\begin{split}\nhs(x)\in & \;\L{(J^{1}(J^{k-1}\vb))_{x},\left(\ext^{n}T^{*}\base\right)_{x}},\\
 & =\L{(J^{1}(\lc^{*}VJ^{k-1}\fb))_{x},\left(\ext^{n}T^{*}\base\right)_{x}},\\
 & =\L{\bigp{(j^{1}\lc)^{*}(VJ^{1}(J^{k-1}\fb))}_{x},\left(\ext^{n}T^{*}\base\right)_{x}},\\
 & =\L{\bigp{VJ^{1}(J^{k-1}\fb)}_{j^{1}\lc(x)},\pi^{1*}\left(\ext^{n}T^{*}\base\right)_{j^{1}\lc(x)}},\\
 & =\L{\bigp{VJ^{1}(J^{k-1}\fb)},\pi^{1*}\left(\ext^{n}T^{*}\base\right)}_{j^{1}\lc(x)},\\
 & =\left[(j^{1}\lc)^{*}\L{\bigp{VJ^{1}(J^{k-1}\fb)},\pi^{1*}\left(\ext^{n}T^{*}\base\right)}\right]_{x}.
\end{split}
\label{eq:Space_of_NonHol_Stresses}
\end{equation}
Here, we view $\pi^{k-1}:J^{k-1}\fb\to\base$ as a fiber bundle over
$\base,$ and so $\pi^{1}:J^{1}(J^{k-1}\fb)\to\base$ is the natural
projection. Hence, values of non-holonomic hyper-stress fields may
be viewed as elements of the vector bundle
\begin{equation}
\pi_{\spnhs}:\spnhs:=\L{\bigp{VJ^{1}(J^{k-1}\fb)},\pi^{1*}\left(\ext^{n}T^{*}\base\right)}\tto J^{1}(J^{k-1}\fb).
\end{equation}
For a given section $\lc:\base\to J^{k-1}\fb$, non-holonomic stresses
at $\lc$ are valued in 
\begin{equation}
(j^{1}\lc)^{*}\spnhs=\L{J^{1}(\lc^{*}VJ^{k-1}\fb),\ext^{n}T^{*}\base}.
\end{equation}

\subsection{Elastic constitutive relations}

Once the spaces of hyper-stresses have been identified, one may introduce
constitutive relations for the variational hyper-stresses and for
the non-holonomic stresses. 

An elastic constitutive relation for the variational hyper-stress
is a section
\begin{equation}
\Psi:J^{k}\fb\tto\spstd=\L{VJ^{k}\fb,\pi^{k*}\ext^{n}T^{*}\base}
\end{equation}
of $\pi_{\spstd}$. Thus, a constitutive relation assigns a value
of a variational hyper-stress at a point $x$ to a $k$-jet of a section
at $x$. In particular, for a section $\conf$, 
\begin{equation}
\Psi\comp j^{k}\conf:\base\tto\spstd
\end{equation}
is identified with a variational hyper-stress field according to (\ref{eq:Space_of_VarStresses}).

Similarly, an elastic constitutive relation for the non-holonomic
stress is a section
\begin{equation}
\Phi:J^{1}(J^{k-1}\fb)\tto\spnhs=\L{\bigp{VJ^{1}(J^{k-1}\fb)},\pi^{1*}\left(\ext^{n}T^{*}\base\right)}
\end{equation}
of $\pi_{\spnhs}$. For a section $\lc:\base\to J^{k-1}\fb$,
\begin{equation}
\Phi\comp j^{1}\lc:\base\tto\spnhs
\end{equation}
is identified with a non-holonomic stress field according to (\ref{eq:Space_of_NonHol_Stresses}).
In particular, for a section $\conf:\base\to\fb$, one has the induced
non-holonomic stress field $\Phi\comp j^{1}(j^{k-1}\conf)$.

\begin{center}
\par\end{center}

\noindent %

\bibliographystyle{alpha}

\end{document}